%% file: main.tex
\documentclass[twocolumn,superscriptaddress, nofootinbib, amsfont,amssymb,amsmath, showpacs,balancelastpage]{revtex4-1}
\usepackage{CJKutf8}
\usepackage[utf8]{inputenc}
\usepackage{amsmath,amssymb,amsfonts,amsthm}
\usepackage{physics}
\usepackage{dsfont}
\usepackage{graphicx}
\usepackage{xcolor}
\usepackage{cancel}
\usepackage{graphicx}
\usepackage{cancel}
\usepackage{color}
\usepackage{bbm}
\usepackage[linktocpage]{hyperref}
\usepackage{appendix}
\hypersetup{colorlinks=true,citecolor=blue,linkcolor=blue, urlcolor=blue, breaklinks=true}
\usepackage[eulergreek]{sansmath}
\usepackage{mathtools}
\usepackage{scalerel}

\usepackage{bm} 
\usepackage{subfigure} 
\usepackage{environ}
\usepackage{url}
\usepackage{hyperref}
\usepackage{adjustbox}
\usepackage{environ}

\usepackage{tikz}

\usetikzlibrary{matrix,arrows.meta,positioning,calc,fit,backgrounds}

\NewEnviron{eqs}{%
\begin{equation}\begin{aligned}
    \BODY
\end{aligned}\end{equation}}

\newcommand{\ZZ}{{\mathbb Z}}
\newcommand{\RR}{{\mathbb R}}

\newcommand{\supp}{\operatorname{supp}}

\newcommand{\eps}{\epsilon}

\newtheorem{theorem}{Theorem}
\newtheorem{corollary}{Corollary}[theorem]
\newtheorem{lemma}[theorem]{Lemma}

\newtheorem{definition}[theorem]{Definition}

\newcommand{\prlsection}[1]{\vspace{6pt}\noindent\textbf{#1.}---}

\graphicspath{{./Figures/}}

\begin{document}

\author{Yitao Feng~\begin{CJK*}{UTF8}{gbsn}(冯逸韬)\end{CJK*}}
\thanks{These authors contributed equally.}
\affiliation{International Center for Quantum Materials, School of Physics, Peking University, Beijing 100871, China}

\author{Hanyu Xue~\begin{CJK*}{UTF8}{gbsn}(薛寒玉)\end{CJK*}}
\thanks{These authors contributed equally.}
\affiliation{Department of Physics, Massachusetts Institute of Technology, Cambridge, Massachusetts 02139, USA}

\author{Yuyang Li~\begin{CJK*}{UTF8}{gbsn}(李雨阳)\end{CJK*}}
\thanks{These authors contributed equally.}
\affiliation{International Center for Quantum Materials, School of Physics, Peking University, Beijing 100871, China}

\author{\\Meng Cheng~\begin{CJK*}{UTF8}{gbsn}(程蒙)\end{CJK*}}
\affiliation{Department of Physics, Yale University, New Haven, CT 06511, USA}

\author{Ryohei Kobayashi~\begin{CJK*}{UTF8}{gbsn}(小林良平)\end{CJK*}}
\affiliation{School of Natural Sciences, Institute for Advanced Study, Princeton, NJ 08540, USA}

\author{Po-Shen Hsin~\begin{CJK*}{UTF8}{bsmi}(辛柏伸)\end{CJK*}}
\affiliation{Department of Mathematics, King’s College London, Strand, London WC2R 2LS, UK}

\author{Yu-An Chen~\begin{CJK*}{UTF8}{bsmi}(陳昱安)\end{CJK*}}
\email[E-mail: ]{yuanchen@pku.edu.cn}
\affiliation{International Center for Quantum Materials, School of Physics, Peking University, Beijing 100871, China}

\date{\today}
\title{Anyonic membranes and Pontryagin statistics}

\begin{abstract}
    Anyons, unique to two spatial dimensions, underlie extraordinary phenomena such as the fractional quantum Hall effect, but their generalization to higher dimensions has remained elusive. The topology of Eilenberg-MacLane spaces constrains the loop statistics to be only bosonic or fermionic in any dimension. In this work, we introduce the novel anyonic statistics for membrane excitations in four dimensions.
    Analogous to the $\mathbb{Z}_N$-particle exhibiting $\mathbb{Z}_{N\times \gcd(2,N)}$ anyonic statistics in two dimensions, we show that the $\mathbb{Z}_N$-membrane possesses $\mathbb{Z}_{N\times \gcd(3,N)}$ anyonic statistics in four dimensions.
    Given unitary volume operators that create membrane excitations on the boundary, we propose an explicit 56-step unitary sequence that detects the membrane statistics.
    We further analyze the boundary theory of $(5{+}1)$D 1-form $\mathbb{Z}_N$ symmetry-protected topological phases and demonstrate that their domain walls realize all possible anyonic membrane statistics.
    We then show that the $\mathbb{Z}_3$ subgroup persists in all higher dimensions.  
    In addition to the standard fermionic $\mathbb{Z}_2$ membrane statistics arising from Stiefel-Whitney classes, membranes also exhibit $\mathbb{Z}_3$ statistics associated with Pontryagin classes.
    We explicitly verify that the 56-step process detects the nontrivial $\mathbb{Z}_3$ statistics in 5, 6, and 7 spatial dimensions. Moreover, in 7 and higher dimensions, the statistics of membrane excitations stabilize to $\mathbb{Z}_{2} \times \mathbb{Z}_{3}$, with the $\mathbb{Z}_3$ sector consistently captured by this process.
\end{abstract}

\maketitle


\prlsection{Introduction}
The quantum statistics of excitations plays a central role across physics, underlying phenomena ranging from superfluids and superconductors to quantum Hall systems and beyond. In particular, condensation is controlled by statistics: bosons can condense into collective quantum states, whereas isolated fermions cannot~\cite{Landau1941Theory, Onsager1949Statistical, Ginzburg1950superconductivity, Bardeen1957Superconductivity, Gu2014supercohomology, Gaiotto2015SpinTQFTs, Bhardwaj:2016clt, Barkeshli:2021ypb}. Two spatial dimensions are special: particle exchange statistics extends beyond bosons and fermions to anyons, which underpin the fractional quantum Hall effect~\cite{Wilczek1982Quantum, Laughlin1983Anomalous, Stormer1999NobelLecture, Moore1991Nonabelions} and provide a route to fault-tolerant quantum computation~\cite{Freedman2002AModular, freedman2003topological, Kitaev2003Fault, Nayak2008NonAbelian}. These facts motivate a basic question: can ``anyon statistics'' be meaningfully generalized beyond two dimensions?

In higher dimensions, excitations can be extended objects such as loops and membranes. For loops, recent lattice definitions and computational frameworks indicate a sharp restriction: in $(3{+}1)$D $\mathbb{Z}_2$ gauge theories relevant to certain superconducting phases and linked to discrete gravitational anomalies~\cite{Thorngren:2014pza, Kobayashi2019gapped, Johnson-Freyd:2020twl, Hsin:2021qiy, CH21, FHH21, Tata2022anomalies}, loop self-statistics can be bosonic or fermionic, giving a $\mathbb{Z}_2$ classification with no higher analogue of anyons~\cite{kobayashi2024generalized, xue2025statistics}.

To move beyond this limitation, we focus on membrane excitations in four and higher spatial dimensions. We refine the algorithm of Ref.~\cite{kobayashi2024generalized} and derive a 56-step unitary sequence that defines a robust statistical invariant for membrane excitations with fusion group $G=\ZZ_2$. This construction is directly analogous to the $T$-junction process that characterizes particle exchange statistics~\cite{Levin2003Fermions}.
Remarkably, the 56-step process applies beyond $\ZZ_2$-membranes, capturing $\ZZ$-membranes\footnote{The fusion group $\ZZ$ means that the generating excitation $a$ has infinite order: $a^n\ne1$ for every nonzero integer $n$, while every excitation has an inverse.} with arbitrary $U(1)$ phases in four dimensions, which we call \emph{anyonic membrane statistics}.
Moreover, this invariant persists in higher dimensions, where it stabilizes into a $\ZZ_3$ structure governed by the first Pontryagin class mod 3, which we refer to as \emph{Pontryagin statistics}, going beyond the familiar $\ZZ_2$ fermionic statistics associated with Stiefel-Whitney classes.

We connect this lattice invariant to higher-form anomalies by constructing explicit realizations on the boundaries of higher-form SPT phases: the membrane excitations arise as symmetry domain walls, and the Berry phase produced by the 56-step process admits a direct interpretation via anomaly inflow. A nontrivial value obstructs a symmetry-preserving short-range-entangled ground state and enforces nontrivial low-energy physics (e.g., symmetry breaking, gaplessness, or topological order), in close analogy with Lieb-Schultz-Mattis-type constraints. From the perspective of quantum error correction, such anomalies and statistics also constrain the realizable logical operations in higher-dimensional topological quantum error-correcting codes. We expect this operational lattice diagnostic to be useful across condensed matter and high-energy theory, as well as quantum information~\cite{Wang2014braiding, Jiang2014Generalized, Jian2014Layer, Wang2015Field, Wang2016Bulk, Ye2016Topological, Pavel2017Braiding, Chan2018Braiding, Wang2019Topological, Zhang2021Compatible, Zhang2022Topological, Barkeshli2023Codimension, Barkeshli2024Higher, Barkeshli2024higherfermion, Sun2025CliffordQCA}.

The paper is organized as follows. We first revisit the lattice $T$-junction process for particle exchange and its relation to the cohomology of Eilenberg-MacLane spaces. We then introduce the 56-step unitary sequence for membrane excitations, prove that it defines a robust statistical invariant, and analyze how its output depends on spatial dimension. Finally, we construct explicit boundary realizations from higher-form SPT phases to demonstrate the predicted membrane statistics and close with a brief discussion of implications and open directions.

\prlsection{Particle statistics from the $T$-junction process} \label{sec:Revisiting the T-Junction process for particle statistics}
To illustrate our approach, we begin with the well-known case of particle statistics.  
In quantum mechanics, particle statistics are determined by the phase acquired when two identical particles are exchanged.  
On a lattice, however, the main challenge is to define such an exchange unambiguously, without being obscured by microscopic details. 

Consider particles placed at the vertices of the lattice $\partial\langle 0123 \rangle$ (the boundary of a 3-simplex):
\begin{equation}
    \vcenter{\hbox{\includegraphics[width=0.3\linewidth]{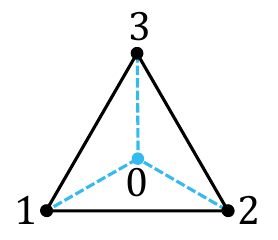}}}~~.
\label{eq:centered-triangle}
\end{equation}  
Throughout this work, we assume that all excitations are \textbf{invertible}, meaning each fuses with its inverse to the vacuum and admits no multiple fusion channels.  
A \textbf{configuration state} $|a\rangle$ on a lattice is specified by a $0$-chain, namely a formal sum of vertices with integer coefficients encoding particle occupancy (negative integers denote antiparticles):  
\begin{equation}
    a := \sum_v a_v \langle v\rangle , \quad a_v \in \ZZ~.
\label{a definition}
\end{equation}
Let $U(s)$ denote the unitary string operator that moves a particle along a path $s$, acting as  
\begin{equation}
    U(s)\,|a\rangle \;\propto\; |a + \partial s\rangle~.
\label{eq:Us_on_a}
\end{equation}  
For example, when $s = e_{ij}$, $U(e_{ij})$ creates a particle at vertex $j$ and an antiparticle at vertex $i$:  
\begin{equation}
    a + \partial s \;=\; a + \langle j \rangle - \langle i \rangle~.
\end{equation}  
The particle exchange statistics are obtained from the sequence of operators~\cite{Levin2003Fermions}:\footnote{For brevity, we denote $U(e_{ij})$ by $U_{ij}$.}
\vspace{-0.25em}
\begin{equation}
    \Theta_T := U_{02} U_{03}^\dagger U_{01} U_{02}^\dagger U_{03} U_{01}^\dagger~.
    \vspace{-0.25em}
\label{eq:T-junction in the introduction}
\end{equation}
This sequence yields an overall phase factor that is independent of the initial state~\cite{FHH21, kobayashi2024generalized, xue2025statistics}.
When $\Theta_T$ acts on the initial state  
$
    \left|\hbox{ \raisebox{-1ex}{\includegraphics[width=.6cm]{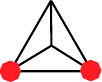}}} \right\rangle,
$  
with particles at vertices 1 and 2, the resulting 6-step $T$-junction process can be visualized as
\vspace{-0.5em}
\begin{equation*}
    \vcenter{\hbox{\includegraphics[width=0.98\linewidth]{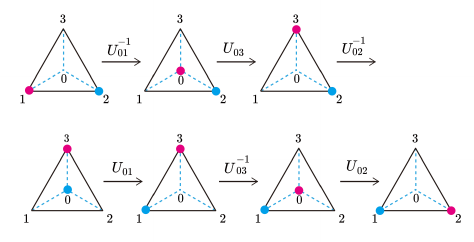}}}.
    \vspace{-1em}
\label{fig:T-junction}
\end{equation*}
This construction reproduces the particle exchange statistics while canceling local phase ambiguities, such as redefining $U_{01} \rightarrow e^{i\phi} U_{01}$.

The key point is that the $T$-junction process $\Theta_T$ is not just a property of the local lattice construction.  
It reflects deeper, global consistency conditions that are naturally captured by a field-theoretic description in one higher dimension.  
More precisely, the possible values of $\Theta_T$ are classified by \emph{invertible topological quantum field theories} (TQFTs), which encode 't~Hooft anomalies~\cite{Liujun2024Symmetriesanomalies, seifnashri2025disentangling, kapustin2025higher, kapustin2025higher2, Kapustin2025Anomalous, kawagoe2025anomaly, Feng2025HigherForm}.

The value of $\Theta_T$ depends on the spatial dimension of the underlying lattice.  
In 2 dimensions, particles can be \emph{anyons}, with exchange statistics $\Theta_T$ taking any $U(1)$ phase factor. 
In three or more spatial dimensions, however, particles are restricted to being either \emph{bosons} ($\Theta_T=1$) or \emph{fermions} ($\Theta_T=-1$).
In $d$ spatial dimensions, particles generate a $(d{-}1)$-form symmetry, and their statistics correspond to an anomaly of this higher-form symmetry.
Mathematically, this can be formulated as
\begin{eqs}
    H^4(B^2 \ZZ, \RR/\ZZ) &= \RR/\ZZ~, &\quad& \text{for } d=2,\\[3pt]
    H^{d+2}(B^{d} \ZZ, \RR/\ZZ) &= \ZZ_{2}~, &\quad& \text{for } d \geq 3,
\label{eq:cohomology of Z-Eilenberg-MacLane (particle) }
\end{eqs}
where $B^n G := K(G,n)$ denotes the Eilenberg-MacLane space~\cite{Eilenberg:1946}, which serves as the classifying space of the $(n{-}1)$-form $G$ symmetry.  

We could also consider $G$-particles with fusion group $G=\ZZ_N$, meaning that $N$ particles fuse to the vacuum.
In this case, the phase of $\Theta_T$ is quantized, corresponding to~\cite{Eilenberg:1946}  
\begin{eqs}
    H^4(B^2\ZZ_N, \RR/\ZZ) &= \ZZ_{N \times \gcd(2,N)}~, \!\!\! & \text{for } d=2,\\[3pt]
    H^{d+2}(B^d\ZZ_N, \RR/\ZZ) &= \ZZ_{\gcd(2,N)}~, & \text{for } d \geq 3,
\label{eq:cohomology of Eilenberg-MacLane (particle)}
\end{eqs}
where $\gcd(2,N)$ denotes the greatest common divisor of $2$ and $N$.
For example, for $\ZZ_2$-particles in 2 spatial dimensions, the anomaly corresponds to $(3{+}1)$D 1-form $\ZZ_2$ symmetry-protected topological (SPT) phases, equivalently described by an invertible TQFT in the $(3{+}1)$D bulk, classified by  
\begin{equation}
    H^4(B^2\ZZ_2, \RR/\ZZ) = \ZZ_4~,
\end{equation}  
which allows exchange statistics $\Theta_T = 1,\, i,\,-1,\,-i$, with $i$ and $-i$ corresponding to the semion and anti-semion, respectively~\cite{Levin2005String, Levin2012Braiding}.  
This matches the classification of Abelian anyon theories in Ref.~\cite{Wang2020InAbelian}, where a quadratic function from $G$ to $\RR/\ZZ$ determines the anyon theory.  
Such quadratic functions are precisely described by $H^4(B^2G, \RR/\ZZ)$~\cite{Eilenberg1953OnTG, Eilenberg1954OnTG, Kapustin2014theta, Kapustin2017Higher, Delcamp2019On}.  

To construct lattice models that realize these anyon statistics in two spatial dimensions, we study the boundary of $(3{+}1)$D 1-form $\ZZ_N$ SPT phases. 
Ref.~\cite{lokman2020lattice} gave lattice realizations of all such phases, corresponding to cocycle representatives in $H^{4}(B^{2}\ZZ_N,\RR/\ZZ)$:
\begin{equation*}
    \begin{cases}
    \dfrac{1}{N}\, B_2 \cup B_2~, & N \not\equiv 0 \pmod{2}~,\\[6pt]
    \dfrac{1}{2N}\!\left(B_2 \cup B_2 + B_2 \cup_1 \delta B_2 \right)~, \!\! & N \equiv 0 \pmod{2}~,
    \end{cases}
\end{equation*}
with explicit boundary constructions.
In the above expression, $B_2$ is a closed cochain modulo $N$. In this and subsequent fractional cocycle expressions, finite-group cochains are represented by integer lifts, with coboundaries taken before reduction and the resulting action valued in $\RR/\ZZ$. Steenrod squares and reduced powers act on the mod-$2$ and mod-$3$ reductions, respectively, whenever the corresponding reduction is defined.
The (generalized) higher cup products are reviewed in SM~Sec.~\ref{app:Higher cup products, May–Steenrod structure, and Steenrod reduce powers}.
It has been shown that particles living on these boundaries exhibit exchange statistics  
\begin{equation}
    \Theta_T =
    \begin{cases}
        \exp\left(\frac{2\pi i}{N}\right)~, & N \not\equiv 0 \pmod{2}~,\\[3pt]
        \exp\left(\frac{2\pi i}{2N}\right)~, & N \equiv 0 \pmod{2}~.
    \end{cases}
\end{equation}
This computation confirms that the $T$-junction process serves as a valid statistical procedure that detects all anomalies $H^4(B^2\ZZ_N, \RR/\ZZ)$ predicted in Eq.~\eqref{eq:cohomology of Eilenberg-MacLane (particle)}.
This analysis extends to higher dimensions~\cite{Chen:2017fvr, Chen:2018nog, Chen2020Exactbosonization}, where only bosons and fermions appear.
The anomaly is given by  
\begin{equation}
    \frac{1}{2} B_d \cup_{d-2} B_d \in H^{d+2}(B^d\mathbb{Z}_N, \mathbb{R}/\mathbb{Z})~, 
\end{equation}
for $N \equiv 0 \pmod{2}$. On a closed oriented $(d+2)$-dimensional bulk spacetime manifold $X_{d+2}$, this has the same evaluated action as $\tfrac{1}{2} w_2 \cup (B_d\bmod 2)$, where $w_2=w_2(TX_{d+2})$ is the second Stiefel-Whitney class~\cite{milnor1974characteristic}.

Loop self-statistics have been studied in Refs.~\cite{FHH21, CH21}. For $\ZZ_N$ fusion in $d\geq3$ spatial dimensions, they admit a $\ZZ_2$ classification for even $N$ and are trivial for odd $N$, corresponding to  
\begin{equation}
    H^{d+2}(B^{d-1}\mathbb{Z}_N,\mathbb{R}/\mathbb{Z}) = \mathbb{Z}_{\gcd(2,N)}, \quad d \geq 3~.
\end{equation}
For even $N$, restrict to backgrounds $B_{d-1}=(N/2)b_{d-1}$, where $b_{d-1}$ is a $\ZZ_2$ cocycle. This restriction to the order-two subgroup determines the nontrivial anomaly class. The restricted class is represented by
\begin{equation*}
    \frac{1}{2} Sq^2Sq^1 b_{d-1}
    =\frac{1}{2}\left(\frac{\delta b_{d-1}}{2}\right)
      \cup_{d-2}\left(\frac{\delta b_{d-1}}{2}\right)~.
\end{equation*}
On a closed oriented $(d+2)$-dimensional manifold, this has the same evaluated action as $\tfrac{1}{2} w_3\cup b_{d-1}$, with $w_3$ the third Stiefel-Whitney class.
On the other hand, since  
\begin{equation}
    H^{d+2}(B^{d-1}\mathbb{Z},\mathbb{R}/\mathbb{Z}) = 0, \quad d\geq2~,
\end{equation}
generic $\mathbb{Z}$-loops admit no nontrivial statistics. This sharply differs from particles, where a single statistical process governs both $G=\mathbb{Z}$ and $\mathbb{Z}_N$.

\prlsection{Membrane statistics from the 56-step process}\label{sec:56-Step unitary process for membrane statistics}
Now, we introduce a 56-step unitary sequence that characterizes membrane statistics in spatial dimensions $d \geq 4$.  
We begin by defining the Hilbert space supporting invertible membrane excitations in 4 spatial dimensions, noting that the construction naturally embeds into higher dimensions.  
Analogous to the triangle with a center vertex in Eq.~\eqref{eq:centered-triangle}, we consider the lattice given by the boundary of a 5-simplex $\langle 012345 \rangle$. 
Configurations of membrane excitations are described by 2-chains
\footnote{Membrane excitations are created by volume operators, so they are actually described by $2$-boundaries.}, i.e., formal sums of faces with integer coefficients:
\begin{equation}
    a = \sum_{f=\langle ijk \rangle} c_f \, \langle ijk \rangle~, 
    \quad c_f \in \ZZ~.
\end{equation}
For each tetrahedron $\langle ijkl \rangle$, we associate a volume operator $U_{ijkl}$ that creates membrane excitations on its four boundary faces:
\begin{equation}
    \langle jkl \rangle - \langle ikl \rangle + \langle ijl \rangle - \langle ijk \rangle.
\end{equation}
From these unitary volume operators, we propose the 56-step unitary sequence
\vspace{-0.5em}
\begin{eqs}
    &\mu_{56}:= \\
    &U_{0235}^\dagger U_{0345}^\dagger U_{0135}^\dagger U_{0134} U_{0245} U_{0134}^\dagger U_{0245}^\dagger U_{0135} \\
    &U_{0125}^\dagger U_{0123} U_{0345} U_{0234}^\dagger U_{0134} U_{0345}^\dagger U_{0234} U_{0134}^\dagger \\
    &U_{0123}^\dagger U_{0125} U_{0245} U_{0123} U_{0134} U_{0135}^\dagger U_{0245}^\dagger U_{0234}^\dagger \\
    &U_{0124}^\dagger U_{0245} U_{0234} U_{0135} U_{0234}^\dagger U_{0245}^\dagger U_{0345} U_{0235} \\
    &U_{0135}^\dagger U_{0345}^\dagger U_{0124} U_{0345} U_{0135} U_{0134}^\dagger U_{0123}^\dagger U_{0234} \\
    &U_{0145}^\dagger U_{0234}^\dagger U_{0123} U_{0124}^\dagger U_{0134} U_{0235}^\dagger U_{0134}^\dagger U_{0124} \\
    &U_{0234} U_{0245} U_{0345}^\dagger U_{0123}^\dagger U_{0345} U_{0245}^\dagger U_{0235} U_{0145}~.
\label{eq:56-step_unitary_process}
\end{eqs}
For illustration, the intermediate states of $\mu_{56}$ acting on the vacuum are presented in SM~Sec.~\ref{app:Intermediate states in the 56-step process}.
This process is obtained by improving the Smith normal form (SNF) algorithm of Ref.~\cite{kobayashi2024generalized},
with computational details given in SM~Sec.~\ref{app:Computational derivation}.

Now we set aside the details of the search algorithm and instead treat $\mu_{56}$ as an explicit candidate process to be verified. We first show that $\mu_{56}$ defines a lattice invariant, robust under local changes of the microscopic operators.
A process satisfying the following two conditions is said to be \textbf{statistical}:
\begin{enumerate}
    \item $\langle a|\mu_{56} |a\rangle$ does not depend on the initial configuration $a$. 

    \item If $U_{ijkl}$ is replaced by another local operator $U_{ijkl}'$ with the same support, satisfying Eq.~\eqref{eq:Us_on_a} and the locality axiom stated in SM~Sec.~\ref{app:Computational derivation}, then $\mu_{56}$ remains unchanged. The configuration axiom allows $U_{ijkl}'$ to differ from $U_{ijkl}$ by a $\ket{a}$-dependent phase factor subject to these locality constraints.
\end{enumerate}
These conditions are equivalent to requiring that \emph{any local process around each vertex cancels out} (also see Refs.~\cite{FHH21, kobayashi2024generalized}).
For example, around a vertex $v$, if applying $U$ changes the state from $a|_v$ to $a'|_v$, then at a later step one must apply $U^{-1}$ to return from $a'|_v$ to $a|_v$.  
Here, $a|_v$ denotes the local configuration in the neighborhood of $v$, which is unaffected by membranes on faces far from $v$.  
As shown in SM~Sec.~\ref{app:Proof of the 56-step unitary as a statistical process}, the 56-step process acting on the vacuum satisfies this \emph{local cancellation criterion}, leading to the following theorem:  
\begin{theorem}
    $\mu_{56}$ is a statistical process for $\ZZ$-membranes in 4 spatial dimensions. 
\end{theorem}
In other words, the process yields a (possibly trivial) global $U(1)$ phase on any initial state, analogous to the $T$-junction process for particle exchange, and remains invariant under local deformations of operators or states.
In higher dimensions, one can embed the $\partial \langle 012345 \rangle$ into the ambient space, and the local cancellation criterion remains automatically satisfied. Moreover, by definition, the cancellation condition for $\ZZ_N$ membranes is strictly weaker than that for $\ZZ$ membranes. We therefore obtain the following corollary:  
\begin{corollary}
    $\mu_{56}$ is a statistical process for $\ZZ$- or $\ZZ_N$-membranes in all spatial dimensions $d \geq 4$.
\end{corollary}

\prlsection{Statistics and anomalies}
Refs.~\cite{kobayashi2024generalized, xue2025statistics} have extended the framework of the $T$-junction process to include higher-dimensional excitations. It is conjectured that for $p$-dimensional excitations with fusion group $G$ in $d$ spatial dimensions, the anomaly is classified by $H^{d+2}(B^{d-p}G,\RR/\ZZ)$. The table of $H^{d+2}(B^{d-2}\ZZ,\RR/\ZZ)$ indicates that the self-statistics of membranes is richer, given by~\cite{Eilenberg:1946} 
\begin{equation}
\begin{aligned}
    H^{6}(B^{2}\ZZ,\RR/\ZZ)     &= \RR/\ZZ~,              &\!\!& \text{for } d=4~, \\[1pt]
    H^{7}(B^{3}\ZZ,\RR/\ZZ)     &= \ZZ_3~,                &\!\!& \text{for } d=5~, \\[1pt]
    H^{8}(B^{4}\ZZ,\RR/\ZZ)     &= \ZZ_3 \times \RR/\ZZ~, &\!\!& \text{for } d=6~, \\[1pt]
    H^{d+2}(B^{d-2}\ZZ,\RR/\ZZ) &= \ZZ_3 \times \ZZ_2,    &\!\!& \text{for } d \ge 7~.
\end{aligned}
\label{eq:cohomology of Z-Eilenberg-MacLane (membrane)}
\end{equation}
For $d \geq 7$, the cohomology stabilizes to $\ZZ_2 \times \ZZ_3$, a consequence of the topology of Eilenberg-MacLane spaces (see, e.g., Theorem~13.2.2 of Ref.~\cite{LodayCyclicHomology}). One can also give a geometric intuition for why the critical dimension $d=7$ appears. The statistical process is constructed from local volume operators supported on the $3$-skeleton $X_3$ (the subcomplex of simplices of dimension $\le 3$). By a classical embedding theorem in simplicial topology, any finite $3$-dimensional simplicial complex embeds into $S^7$~\cite{MunkresTopology}. Hence, once $d\ge 7$, any local configuration relevant to membrane statistics can be realized inside a triangulated $S^d$, and no essentially new local patterns appear as $d$ increases. This provides intuition for why the membrane statistics stabilizes for $d\ge 7$ (more generally, one expects stabilization for $d\ge 2p+3$ for $p$-dimensional excitations).

In the rest of this paper, we prove that the 56-step process $\mu_{56}$ detects the $\RR/\ZZ \cong U(1)$ invariant in 4 dimensions and the $\ZZ_3$ sector in higher dimensions.
So far, we have shown that $\mu_{56}$ is a statistical process, but it could still be trivial. To establish nontriviality of $\mu_{56}$, explicit lattice models are constructed: specifying the operators $U_t$ and substituting them into the sequence shows that $\mu_{56}$ yields a nontrivial phase, e.g., $e^{2\pi i/3}$ in 5 or higher spatial dimensions. 
To identify such models, we examine the anomaly associated with $\ZZ_N$ membranes.
Analogous to Eq.~\eqref{eq:cohomology of Eilenberg-MacLane (particle)}, the relevant cohomology groups of Eilenberg-MacLane spaces are:
\begin{equation}
\begin{aligned}
    H^{6}(B^{2}\ZZ_N,\RR/\ZZ) &= {\color{red}\ZZ_{N\times \gcd(3,N)}}~, \\[2pt]
    H^{7}(B^{3}\ZZ_N,\RR/\ZZ) &= {\color{red}\ZZ_{\gcd(3,N)}} \times {\color{blue}\ZZ_N}~, \\[2pt]
    H^{8}(B^{4}\ZZ_N,\RR/\ZZ) &= {\color{red}\ZZ_{\gcd(3,N)}} \times {\color{blue}\ZZ_{N\times \gcd(2,N)}}~, \\[2pt]
    H^{d+2}(B^{d-2}\ZZ_N,\RR/\ZZ) &= {\color{red}\ZZ_{\gcd(3,N)}} \times {\color{blue}\ZZ_{\gcd(2,N)}}~,
\end{aligned}
\label{eq:cohomology of Eilenberg-MacLane (membrane)}
\end{equation}
where the last line applies for $d\ge 7$. We will construct explicit lattice models (with additional details in SM~Sec.~\ref{app: Membrane statistics in d>=4 spatial dimensions}) and establish the following theorem:
\begin{theorem}
The 56-step unitary process defined in Eq.~\eqref{eq:56-step_unitary_process} detects the anomaly classes marked in red in Eq.~\eqref{eq:cohomology of Eilenberg-MacLane (membrane)}. For fusion group $\ZZ_N$, its possible outputs are
\begin{equation*}
    \mu_{56} \in
    \begin{cases}
        \{z\in U(1)\mid z^{N\gcd(3,N)}=1\}~, & d=4~, \\[3pt]
        \{z\in U(1)\mid z^{\gcd(3,N)}=1\}~, & d\geq5~.
    \end{cases}
\end{equation*}
For fusion group $\ZZ$, the corresponding output ranges are $U(1)$ for $d=4$ and $\{e^{2\pi i p/3}\mid p=0,1,2\}$ for $d\geq5$.
\end{theorem}
\noindent This result is analogous to particles in two dimensions, which exhibit $U(1)$ exchange statistics, while in 3 or higher dimensions the statistics reduce to $\ZZ_2$.  
In contrast, $\mu_{56}$ stabilizes into $\ZZ_3$ statistics rather than $\ZZ_2$.
As we will show, this distinction originates from the difference between the Steenrod square (defined mod 2) and the Steenrod reduced power (defined mod 3).\footnote{
Our 56-step process does not detect the $\ZZ_2$ sector in Eq.~\eqref{eq:cohomology of Eilenberg-MacLane (membrane)}. This $\ZZ_2$ sector is associated with the anomaly
$\tfrac12\operatorname{Sq}^4 B_{d-2}=\tfrac12 B_{d-2}\cup_{d-6} B_{d-2}$ (with $B_{d-2}$ reduced modulo $2$)~\cite{kobayashi2024generalized}.
Since this term is quadratic in $B_{d-2}$, it can be realized within the Pauli stabilizer framework~\cite{Sun2025CliffordQCA}. However, Lemma~\ref{lemma: mu56 = 1 for Pauli} shows that $\mu_{56}=1$ whenever all volume operators are Pauli operators, and therefore $\mu_{56}$ cannot probe this $\ZZ_2$ sector. In follow-up work, we construct an alternative statistical process that detects this $\ZZ_2$ sector within the Pauli framework~\cite{Feng2025PauliTQFT}.
}
The relations between the cohomology classes in Eqs.~\eqref{eq:cohomology of Z-Eilenberg-MacLane (membrane)} and \eqref{eq:cohomology of Eilenberg-MacLane (membrane)} across different dimensions and fusion groups are summarized in SM~Sec.~\ref{app:Fusion group change and dimensional reduction}.

\begin{table*}[t]
\centering
\renewcommand{\arraystretch}{1.5}
\begin{tabular}{c|c|c|c} 
\hline
Dimension $d$ & SPT action & $\mu_{56}$ & Statistics group  \\ 
\hline\hline
$4$         & $\frac{1}{N}B_2\cup B_2\cup B_2$           &    $e^{2\pi i/N}$      &            $\ZZ_N\subseteq {\color{red}\ZZ_{N\times \gcd{(3,N)}}}$    \\ 
\hline
$4$      &$\frac{1}{9} \left(B_2 \cup B_2 \cup B_2 + (B_2 \cup_1 \delta B_2) \cup B_2 - B_2 \cup (B_2 \cup_1 \delta B_2) \right)$ &   ${\color{red}e^{2\pi i/9}}$       &       ${\color{red}\ZZ_{9}}$            \\ 
\hline
$5$    & Steenrod reduced power $\frac{1}{3} P^1(B_3)$ in SM~Eq.~\eqref{eq:P1(B3)}       &     ${\color{red}e^{2\pi i/3}}$     &               ${\color{red}\ZZ_3} \times {\color{blue}\ZZ_3}$    \\ 
\hline
$6$    &  Steenrod reduced power $\frac{1}{3} P^1(B_4)$ in SM~Eq.~\eqref{eq:P1(B4)}      &    ${\color{red}e^{2\pi i/3}}$      &              ${\color{red}\ZZ_3} \times {\color{blue}\ZZ_3}$     \\ 
\hline
$7$    &  Steenrod reduced power  $\frac{1}{3} P^1(B_5)$ in SM~Eq.~\eqref{eq:P1(B5)}     &     ${\color{red}e^{2\pi i/3}}$     &               ${\color{red}\ZZ_3}$    \\
\hline
\end{tabular}
\caption{The output of the 56-step process for different SPT actions, together with their statistics $\mu_{56}$. We have verified that the process is able to detect the statistics of the higher Pontryagin power in four spatial dimensions, and the Steenrod reduced power in higher dimensions. The first row shows the result verified for arbitrary $N$ by substituting $\ZZ$-valued $B_2$ into the algorithm, while the remaining rows illustrate the representative case of $\ZZ_3$ membranes.
In the table, only the statistics indicated in red are detected by the 56-step process.}
\label{table: output mu56}
\end{table*}

We first focus on membrane statistics in $d=4$ by realizing symmetry domain-wall membranes on the $(4{+}1)$D boundary of a $(5{+}1)$D 1-form $\ZZ_N$ SPT phase. We will construct the boundary Hamiltonian and operators that transport membrane excitations and evaluate the 56-step process $\mu_{56}$. The resulting Berry phase matches Eq.~\eqref{eq:cohomology of Eilenberg-MacLane (membrane)}, showing that $\mu_{56}$ detects anyonic membrane statistics in four dimensions. 

Consider $\ZZ_N$ membrane excitations in four spatial dimensions. We describe a membrane by a closed $2$-cochain $B_2\in Z^2(M_5,\ZZ_N)$ on the spacetime manifold $M_5$, Poincar\'e dual to the membrane worldvolume.
In this setting, the anomaly is captured by a $(5{+}1)$D 1-form $\ZZ_N$ SPT phase classified by $H^6(B^2\ZZ_N,\RR/\ZZ)$~\cite{kobayashi2024generalized}. A convenient choice of cocycle representatives of $H^6(B^2\ZZ_N, \RR/\ZZ)$ is
\begin{equation}
    {
    \begin{cases}
    {\color{red}\dfrac{1}{N} B_2 \cup B_2 \cup B_2}~, 
    &\!\!\! N \not\equiv 0 \pmod{3}~,\\[3pt]
    {\color{red}\dfrac{1}{3N}
    \!\left[\begin{aligned}
     &~~B_2 \cup B_2 \cup B_2 \\
    &+ (B_2 \cup_1 \delta B_2)\cup B_2 \\
    &-~B_2 \cup (B_2 \cup_1 \delta B_2)
    \end{aligned}~\right]
    }
    , 
    & \!\!\! N \equiv 0 \pmod{3}~,
    \end{cases}
    }
\label{eq:higher Pontryagin power}
\end{equation}
where the second line defines the \emph{higher Pontryagin power}~\cite{hsin2024classifying}. In SM~Sec.~\ref{app: Membrane statistics in d>=4 spatial dimensions}, we construct the corresponding boundary theory and show that the domain-wall membranes exhibit
\begin{equation*}
    \mu_{56} =
    \begin{cases}
        \exp\!\left(\frac{2\pi i}{N}\right)~, & N\not\equiv 0 \pmod{3}~,\\[3pt]
        \exp\!\left(\frac{2\pi i}{3N}\right)~, & N\equiv 0 \pmod{3}~,
    \end{cases}
\end{equation*}
thereby demonstrating that $\mu_{56}$ is a nontrivial statistical invariant specified by the underlying SPT anomaly.

We emphasize that the mod-$3$ part of the above cocycle is related to the first Pontryagin class. On a closed oriented six-manifold $M_6$, for $B_2\in H^2(M_6,\ZZ_3)$ and $p_1=p_1(TM_6)\in H^4(M_6,\ZZ)$, the mod-$3$ Steenrod reduced power obeys the top-degree cohomology relation~\cite{Tomonaga1965PontryaginMod3}
\begin{equation}
    B_2 \cup B_2 \cup B_2 \equiv p_1 \cup B_2 \pmod{3}~.
\end{equation}
Thus, the $\ZZ_3$ sector is controlled by $p_1 \bmod 3$, and we refer to it as \emph{Pontryagin statistics}.

Similar analyses for $d>4$ are presented in SM~Sec.~\ref{app: Membrane statistics in d>=4 spatial dimensions}, where we explicitly construct lattice models for the corresponding SPT phases and show that their boundaries support domain-wall membrane excitations with nontrivial statistics. The numerical values of $\mu_{56}$ for $d=4,5,6,7$ are summarized in Table~\ref{table: output mu56}. The algorithm for evaluating $\mu_{56}$ from a given cocycle is described in SM~Sec.~\ref{sec:Algorithm for verifying membrane statistics}.

\prlsection{Discussions}\label{sec:Discussions}
In this work, we introduced a 56-step unitary process that defines and detects membrane statistics in four and higher spatial dimensions. Starting from four dimensions, the process captures $U(1)$ anyonic membrane statistics, which persist in higher dimensions and stabilize to a $\ZZ_3$ classification governed by the first Pontryagin class mod 3. This provides a unifying framework for generalized membrane statistics across dimensions.
Several important questions remain open. One is to clarify the precise geometric meaning of the 56-step process. Another is to identify additional statistical processes of membrane excitations, particularly for the $\ZZ_2$ family. Resolving these issues would deepen the connection between lattice constructions, anomalies of higher symmetries, and cohomological operations. 

Finally, our findings are consistent with the global framing structure captured by the homotopy group $\pi_{d+1}(S^{d-2})$, which classifies framed cobordism classes of membrane worldvolumes in $(d+1)$-dimensional spacetime via the Pontryagin–Thom construction~\cite{Gukov:2025dol}.
In this language, a correlation function that detects a change of framing is naturally interpreted as an obstruction to gauging the symmetry generated by membrane-creation operators, i.e., to condensing the membrane excitations. Since our lattice construction realizes only a restricted set of membrane operators (e.g., fixed by a given fusion rule), the resulting anomaly invariants can capture only the corresponding subgroup of the full framing classification.
For $d\ge 5$, $\pi_{d+1}(S^{d-2})$ contains a $\ZZ_{12}$ subgroup (and it is $\mathbb{Z}_{24}$ for $d\geq 7$), consistent with the stabilized $\ZZ_{\gcd(2,N)}\times \ZZ_{\gcd(3,N)}$ structure revealed by our construction~\cite{Gukov:2025dol}. An interesting open direction is to clarify the relation between this $\ZZ_{12}$ structure and the $\ZZ_2\times\ZZ_3\cong\ZZ_6$ cohomological classification considered here, and to determine whether the full $\ZZ_{12}$ (or $\mathbb{Z}_{24}$ for $d\geq 7$) can be detected by lattice anomaly diagnostics.

\section*{Acknowledgment}

We thank Xie Chen, Lukasz Fidkowski, Jeongwan Haah, Anton Kapustin, Yuji Tachikawa, Qing-Rui Wang, and Xiao-Gang Wen for insightful discussions. We are also grateful to Anibal M. Medina-Mardones for clarifying the cochain-level May–Steenrod operations and for kindly sharing the Python notebook.

Y.-A.C. is supported by the National Natural Science Foundation of China (Grant No.~12474491) and the Fundamental Research Funds for the Central Universities, Peking University.
R.K. is supported by the U.S. Department of Energy through grant number DE-SC0009988 and the Sivian Fund.
M.C. is supported by NSF grant DMR-2424315.
P.-S.H. is supported by the Department of Mathematics, King's College London.
H.X. is supported by the Department of Physics, Massachusetts Institute of Technology.

\bibliography{bibliography}

\input{supp_mat.tex}

\vfill

\end{document}

%% file: supp_mat.tex
\onecolumngrid
\newpage
\begin{center}
\Large{\bf Supplemental Materials}
\end{center}


\section{Intermediate states in the 56-step process}\label{app:Intermediate states in the 56-step process}

When we apply the 56-step process to the initial vacuum state, all intermediate states are:
{\footnotesize
\begin{equation*}
\begin{split}
U_{0145} \;&\to\; -\langle 014\rangle + \langle 015\rangle - \langle 045\rangle + \langle 145\rangle \\
U_{0235} \;&\to\; -\langle 014\rangle + \langle 015\rangle - \langle 023\rangle + \langle 025\rangle - \langle 035\rangle - \langle 045\rangle + \langle 145\rangle + \langle 235\rangle \\
U_{0245}^\dagger \;&\to\; -\langle 014\rangle + \langle 015\rangle - \langle 023\rangle + \langle 024\rangle - \langle 035\rangle + \langle 145\rangle + \langle 235\rangle - \langle 245\rangle \\
U_{0345} \;&\to\; -\langle 014\rangle + \langle 015\rangle - \langle 023\rangle + \langle 024\rangle - \langle 034\rangle - \langle 045\rangle + \langle 145\rangle + \langle 235\rangle - \langle 245\rangle + \langle 345\rangle \\
U_{0123}^\dagger \;&\to\; \langle 012\rangle - \langle 013\rangle - \langle 014\rangle + \langle 015\rangle + \langle 024\rangle - \langle 034\rangle - \langle 045\rangle - \langle 123\rangle + \langle 145\rangle + \langle 235\rangle - \langle 245\rangle + \langle 345\rangle \\
U_{0345}^\dagger \;&\to\; \langle 012\rangle - \langle 013\rangle - \langle 014\rangle + \langle 015\rangle + \langle 024\rangle - \langle 035\rangle - \langle 123\rangle + \langle 145\rangle + \langle 235\rangle - \langle 245\rangle \\
U_{0245} \;&\to\; \langle 012\rangle - \langle 013\rangle - \langle 014\rangle + \langle 015\rangle + \langle 025\rangle - \langle 035\rangle - \langle 045\rangle - \langle 123\rangle + \langle 145\rangle + \langle 235\rangle \\
U_{0234} \;&\to\; \langle 012\rangle - \langle 013\rangle - \langle 014\rangle + \langle 015\rangle - \langle 023\rangle + \langle 024\rangle + \langle 025\rangle - \langle 034\rangle - \langle 035\rangle - \langle 045\rangle - \langle 123\rangle + \langle 145\rangle + \langle 234\rangle + \langle 235\rangle \\
U_{0124} \;&\to\; -\langle 013\rangle + \langle 015\rangle - \langle 023\rangle + \langle 025\rangle - \langle 034\rangle - \langle 035\rangle - \langle 045\rangle - \langle 123\rangle + \langle 124\rangle + \langle 145\rangle + \langle 234\rangle + \langle 235\rangle \\
U_{0134}^\dagger \;&\to\; -\langle 014\rangle + \langle 015\rangle - \langle 023\rangle + \langle 025\rangle - \langle 035\rangle - \langle 045\rangle - \langle 123\rangle + \langle 124\rangle - \langle 134\rangle + \langle 145\rangle + \langle 234\rangle + \langle 235\rangle \\
U_{0235}^\dagger \;&\to\; -\langle 014\rangle + \langle 015\rangle - \langle 045\rangle - \langle 123\rangle + \langle 124\rangle - \langle 134\rangle + \langle 145\rangle + \langle 234\rangle \\
U_{0134} \;&\to\; -\langle 013\rangle + \langle 015\rangle - \langle 034\rangle - \langle 045\rangle - \langle 123\rangle + \langle 124\rangle + \langle 145\rangle + \langle 234\rangle \\
U_{0124}^\dagger \;&\to\; \langle 012\rangle - \langle 013\rangle - \langle 014\rangle + \langle 015\rangle + \langle 024\rangle - \langle 034\rangle - \langle 045\rangle - \langle 123\rangle + \langle 145\rangle + \langle 234\rangle \\
U_{0123} \;&\to\; -\langle 014\rangle + \langle 015\rangle - \langle 023\rangle + \langle 024\rangle - \langle 034\rangle - \langle 045\rangle + \langle 145\rangle + \langle 234\rangle \\
U_{0234}^\dagger \;&\to\; -\langle 014\rangle + \langle 015\rangle - \langle 045\rangle + \langle 145\rangle \\
U_{0145}^\dagger \;&\to\; 0 \\
U_{0234} \;&\to\; -\langle 023\rangle + \langle 024\rangle - \langle 034\rangle + \langle 234\rangle \\
U_{0123}^\dagger \;&\to\; \langle 012\rangle - \langle 013\rangle + \langle 024\rangle - \langle 034\rangle - \langle 123\rangle + \langle 234\rangle \\
U_{0134}^\dagger \;&\to\; \langle 012\rangle - \langle 014\rangle + \langle 024\rangle - \langle 123\rangle - \langle 134\rangle + \langle 234\rangle \\
U_{0135} \;&\to\; \langle 012\rangle - \langle 013\rangle - \langle 014\rangle + \langle 015\rangle + \langle 024\rangle - \langle 035\rangle - \langle 123\rangle - \langle 134\rangle + \langle 135\rangle + \langle 234\rangle \\
U_{0345} \;&\to\; \langle 012\rangle - \langle 013\rangle - \langle 014\rangle + \langle 015\rangle + \langle 024\rangle - \langle 034\rangle - \langle 045\rangle - \langle 123\rangle - \langle 134\rangle + \langle 135\rangle + \langle 234\rangle + \langle 345\rangle \\
U_{0124} \;&\to\; -\langle 013\rangle + \langle 015\rangle - \langle 034\rangle - \langle 045\rangle - \langle 123\rangle + \langle 124\rangle - \langle 134\rangle + \langle 135\rangle + \langle 234\rangle + \langle 345\rangle \\
U_{0345}^\dagger \;&\to\; -\langle 013\rangle + \langle 015\rangle - \langle 035\rangle - \langle 123\rangle + \langle 124\rangle - \langle 134\rangle + \langle 135\rangle + \langle 234\rangle \\
U_{0135}^\dagger \;&\to\; -\langle 123\rangle + \langle 124\rangle - \langle 134\rangle + \langle 234\rangle \\
U_{0235} \;&\to\; -\langle 023\rangle + \langle 025\rangle - \langle 035\rangle - \langle 123\rangle + \langle 124\rangle - \langle 134\rangle + \langle 234\rangle + \langle 235\rangle \\
U_{0345} \;&\to\; -\langle 023\rangle + \langle 025\rangle - \langle 034\rangle - \langle 045\rangle - \langle 123\rangle + \langle 124\rangle - \langle 134\rangle + \langle 234\rangle + \langle 235\rangle + \langle 345\rangle \\
U_{0245}^\dagger \;&\to\; -\langle 023\rangle + \langle 024\rangle - \langle 034\rangle - \langle 123\rangle + \langle 124\rangle - \langle 134\rangle + \langle 234\rangle + \langle 235\rangle - \langle 245\rangle + \langle 345\rangle \\
U_{0234}^\dagger \;&\to\; -\langle 123\rangle + \langle 124\rangle - \langle 134\rangle + \langle 235\rangle - \langle 245\rangle + \langle 345\rangle \\
U_{0135} \;&\to\; -\langle 013\rangle + \langle 015\rangle - \langle 035\rangle - \langle 123\rangle + \langle 124\rangle - \langle 134\rangle + \langle 135\rangle + \langle 235\rangle - \langle 245\rangle + \langle 345\rangle \\
U_{0234} \;&\to\; -\langle 013\rangle + \langle 015\rangle - \langle 023\rangle + \langle 024\rangle - \langle 034\rangle - \langle 035\rangle - \langle 123\rangle + \langle 124\rangle - \langle 134\rangle + \langle 135\rangle + \langle 234\rangle + \langle 235\rangle - \langle 245\rangle + \langle 345\rangle \\
U_{0245} \;&\to\; -\langle 013\rangle + \langle 015\rangle - \langle 023\rangle + \langle 025\rangle - \langle 034\rangle - \langle 035\rangle - \langle 045\rangle - \langle 123\rangle + \langle 124\rangle - \langle 134\rangle + \langle 135\rangle + \langle 234\rangle + \langle 235\rangle + \langle 345\rangle \\
U_{0124}^\dagger \;&\to\; \langle 012\rangle - \langle 013\rangle - \langle 014\rangle + \langle 015\rangle - \langle 023\rangle + \langle 024\rangle + \langle 025\rangle - \langle 034\rangle - \langle 035\rangle - \langle 045\rangle - \langle 123\rangle - \langle 134\rangle + \langle 135\rangle + \langle 234\rangle + \langle 235\rangle + \langle 345\rangle \\
U_{0234}^\dagger \;&\to\; \langle 012\rangle - \langle 013\rangle - \langle 014\rangle + \langle 015\rangle + \langle 025\rangle - \langle 035\rangle - \langle 045\rangle - \langle 123\rangle - \langle 134\rangle + \langle 135\rangle + \langle 235\rangle + \langle 345\rangle \\
U_{0245}^\dagger \;&\to\; \langle 012\rangle - \langle 013\rangle - \langle 014\rangle + \langle 015\rangle + \langle 024\rangle - \langle 035\rangle - \langle 123\rangle - \langle 134\rangle + \langle 135\rangle + \langle 235\rangle - \langle 245\rangle + \langle 345\rangle \\
U_{0135}^\dagger \;&\to\; \langle 012\rangle - \langle 014\rangle + \langle 024\rangle - \langle 123\rangle - \langle 134\rangle + \langle 235\rangle - \langle 245\rangle + \langle 345\rangle \\
U_{0134} \;&\to\; \langle 012\rangle - \langle 013\rangle + \langle 024\rangle - \langle 034\rangle - \langle 123\rangle + \langle 235\rangle - \langle 245\rangle + \langle 345\rangle \\
U_{0123} \;&\to\; -\langle 023\rangle + \langle 024\rangle - \langle 034\rangle + \langle 235\rangle - \langle 245\rangle + \langle 345\rangle \\
U_{0245} \;&\to\; -\langle 023\rangle + \langle 025\rangle - \langle 034\rangle - \langle 045\rangle + \langle 235\rangle + \langle 345\rangle \\
U_{0125} \;&\to\; -\langle 012\rangle + \langle 015\rangle - \langle 023\rangle - \langle 034\rangle - \langle 045\rangle + \langle 125\rangle + \langle 235\rangle + \langle 345\rangle \\
U_{0123}^\dagger \;&\to\; -\langle 013\rangle + \langle 015\rangle - \langle 034\rangle - \langle 045\rangle - \langle 123\rangle + \langle 125\rangle + \langle 235\rangle + \langle 345\rangle \\
\end{split}
\end{equation*}
\begin{equation*}
\begin{split}
U_{0134}^\dagger \;&\to\; -\langle 014\rangle + \langle 015\rangle - \langle 045\rangle - \langle 123\rangle + \langle 125\rangle - \langle 134\rangle + \langle 235\rangle + \langle 345\rangle \\
U_{0234} \;&\to\; -\langle 014\rangle + \langle 015\rangle - \langle 023\rangle + \langle 024\rangle - \langle 034\rangle - \langle 045\rangle - \langle 123\rangle + \langle 125\rangle - \langle 134\rangle + \langle 234\rangle + \langle 235\rangle + \langle 345\rangle \\
U_{0345}^\dagger \;&\to\; -\langle 014\rangle + \langle 015\rangle - \langle 023\rangle + \langle 024\rangle - \langle 035\rangle - \langle 123\rangle + \langle 125\rangle - \langle 134\rangle + \langle 234\rangle + \langle 235\rangle \\
U_{0134} \;&\to\; -\langle 013\rangle + \langle 015\rangle - \langle 023\rangle + \langle 024\rangle - \langle 034\rangle - \langle 035\rangle - \langle 123\rangle + \langle 125\rangle + \langle 234\rangle + \langle 235\rangle \\
U_{0234}^\dagger \;&\to\; -\langle 013\rangle + \langle 015\rangle - \langle 035\rangle - \langle 123\rangle + \langle 125\rangle + \langle 235\rangle \\
U_{0345} \;&\to\; -\langle 013\rangle + \langle 015\rangle - \langle 034\rangle - \langle 045\rangle - \langle 123\rangle + \langle 125\rangle + \langle 235\rangle + \langle 345\rangle \\
U_{0123} \;&\to\; -\langle 012\rangle + \langle 015\rangle - \langle 023\rangle - \langle 034\rangle - \langle 045\rangle + \langle 125\rangle + \langle 235\rangle + \langle 345\rangle \\
U_{0125}^\dagger \;&\to\; -\langle 023\rangle + \langle 025\rangle - \langle 034\rangle - \langle 045\rangle + \langle 235\rangle + \langle 345\rangle \\
U_{0135} \;&\to\; -\langle 013\rangle + \langle 015\rangle - \langle 023\rangle + \langle 025\rangle - \langle 034\rangle - \langle 035\rangle - \langle 045\rangle + \langle 135\rangle + \langle 235\rangle + \langle 345\rangle \\
U_{0245}^\dagger \;&\to\; -\langle 013\rangle + \langle 015\rangle - \langle 023\rangle + \langle 024\rangle - \langle 034\rangle - \langle 035\rangle + \langle 135\rangle + \langle 235\rangle - \langle 245\rangle + \langle 345\rangle \\
U_{0134}^\dagger \;&\to\; -\langle 014\rangle + \langle 015\rangle - \langle 023\rangle + \langle 024\rangle - \langle 035\rangle - \langle 134\rangle + \langle 135\rangle + \langle 235\rangle - \langle 245\rangle + \langle 345\rangle \\
U_{0245} \;&\to\; -\langle 014\rangle + \langle 015\rangle - \langle 023\rangle + \langle 025\rangle - \langle 035\rangle - \langle 045\rangle - \langle 134\rangle + \langle 135\rangle + \langle 235\rangle + \langle 345\rangle \\
U_{0134} \;&\to\; -\langle 013\rangle + \langle 015\rangle - \langle 023\rangle + \langle 025\rangle - \langle 034\rangle - \langle 035\rangle - \langle 045\rangle + \langle 135\rangle + \langle 235\rangle + \langle 345\rangle \\
U_{0135}^\dagger \;&\to\; -\langle 023\rangle + \langle 025\rangle - \langle 034\rangle - \langle 045\rangle + \langle 235\rangle + \langle 345\rangle \\
U_{0345}^\dagger \;&\to\; -\langle 023\rangle + \langle 025\rangle - \langle 035\rangle + \langle 235\rangle \\
U_{0235}^\dagger \;&\to\; 0
\end{split}
\end{equation*}
}

\section{Proof of the 56-step unitary as a statistical process}\label{app:Proof of the 56-step unitary as a statistical process}

In this section, we prove that the 56-step unitary defined in Eq.~\eqref{eq:56-step_unitary_process} is a \textbf{statistical process}, independent of the particular choice of the volume operator $U_{ijkl}$, provided it satisfies Eq.~\eqref{eq:Us_on_a}. As discussed in the main text, it suffices to verify the \textbf{local cancellation criterion}~\cite{FHH21, xue2025statistics}. Roughly speaking, whenever a local configuration change occurs during the process, a corresponding reverse operation must appear later so that any phase accumulated from local interactions cancels exactly.

To formulate the local cancellation criterion more precisely, we follow Refs.~\cite{kobayashi2024generalized, xue2025statistics}. We first define the phase angle $\theta(U(s),a)$ by the action of a unitary operator $U(s)$ on a state $\lvert a\rangle$:  
\begin{equation}
U(s)\,\lvert a\rangle \;=\; \exp \!\big( i\theta(U(s),a) \big)\, \lvert a + \partial s\rangle ,
\label{eq:theta_def}
\end{equation}
where $\theta(U(s),a)$ depends on the operator $U(s)$ and the state $\lvert a\rangle$. The set $\{\theta(U,a)\}$ for all $U$ and $a$ encodes the statistical information. These phases are not independent; for instance,  
\begin{equation}
    \lvert a\rangle =U(s)U(s)^\dagger\,\lvert a\rangle \;=\; \exp \!\big(i\theta(U(s)^\dagger,a)\big)
    \exp \!\big(i\theta(U(s),a - \partial s)\big)\lvert a \rangle ,
\end{equation}
which implies
\begin{equation}
    \theta(U(s)^\dagger,a) = -\theta(U(s),a - \partial s) \pmod{2\pi}~.
\label{eq:theta_dagger_def}
\end{equation}
As an illustration, consider the $T$-junction process in Eq.~\eqref{eq:T-junction in the introduction}.
With two particles at vertices 1 and 2 as the initial state (represented by $a=\langle 1\rangle+\langle 2\rangle$, see Eq.~\eqref{a definition}), the total accumulated phase is
\begin{eqs}
    &\theta( U_{02} U_{03}^\dagger U_{01} U_{02}^\dagger U_{03} U_{01}^\dagger,\langle 1\rangle + \langle 2\rangle)\\
    =&~\theta( U_{01}^\dagger,\langle 1\rangle + \langle 2\rangle) 
    + \theta( U_{03},\langle 0\rangle + \langle 2\rangle) 
    + \theta( U_{02}^\dagger,\langle 2\rangle + \langle 3\rangle)\\
    &+ \theta( U_{01},\langle 0\rangle + \langle 3\rangle)
    + \theta( U_{03}^\dagger,\langle 1\rangle + \langle 3\rangle)
    + \theta( U_{02},\langle 0\rangle + \langle 1\rangle)\\
    =&- \theta( U_{01},\langle 0\rangle + \langle 2\rangle) 
    + \theta( U_{03},\langle 0\rangle + \langle 2\rangle) 
    - \theta( U_{02},\langle 0\rangle + \langle 3\rangle)\\
    &+ \theta( U_{01},\langle 0\rangle + \langle 3\rangle)
    - \theta( U_{03},\langle 0\rangle + \langle 1\rangle)
    + \theta( U_{02},\langle 0\rangle + \langle 1\rangle)~,
\end{eqs}
where
Eq.~\eqref{eq:theta_dagger_def} has been used to express all terms in terms of $U$ without $U^\dagger$.

We now define the projection
\begin{equation}
    a \;\rightarrow\; a|_v,
\end{equation}
which keeps only those simplices in $a$ that contain $v$ as a vertex.\footnote{More formally, this projection identifies $a$ and $a'$ if $a-a'$ is generated by excitation operators away from $v$; see Theorem~III.3 of Ref.~\cite{xue2025statistics}.}
For the phase angle $\theta(U,a)$, we introduce the corresponding projection
\begin{equation}
    P_v \,\theta(U(s),a) :=
    \begin{cases}
        \theta(U(s), a|_v), & \text{if $v \in s$,}\\
        0, & \text{otherwise,}
    \end{cases}
\end{equation}
where $v \in s$ indicates that $v$ is a vertex of the simplex $s$.  
With these definitions, the local cancellation condition for a process $V=\prod U^\pm$ can be stated precisely as
\begin{equation}
    P_v\,\theta(V,a) = 0, \quad \forall\, v~.
\end{equation}
It can be verified that the $T$-junction process satisfies the local cancellation criterion.  
For example, choosing $v=\langle 0 \rangle$ gives
\begin{eqs}
    &P_0 \,\theta( U_{02} U_{03}^\dagger U_{01} U_{02}^\dagger U_{03} U_{01}^\dagger,\langle 1\rangle + \langle 2\rangle)\\
    =&-P_0\,\theta ( U_{01},\langle 0\rangle + \langle 2\rangle ) 
    + P_0\,\theta ( U_{03},\langle 0\rangle + \langle 2\rangle) 
    - P_0\,\theta ( U_{02},\langle 0\rangle + \langle 3\rangle )\\
    &+ P_0\,\theta ( U_{01},\langle 0\rangle + \langle 3\rangle )
    - P_0\,\theta ( U_{03},\langle 0\rangle + \langle 1\rangle)
    + P_0\,\theta ( U_{02},\langle 0\rangle + \langle 1\rangle )\\
    =&- \theta ( U_{01},\langle 0\rangle ) 
    + \theta ( U_{03},\langle 0\rangle ) 
    - \theta ( U_{02},\langle 0\rangle )
    + \theta ( U_{01},\langle 0\rangle )
    - \theta ( U_{03},\langle 0\rangle )
    + \theta ( U_{02},\langle 0\rangle )
    = 0~,
\end{eqs}
while choosing $v=\langle 1 \rangle$ yields
\begin{eqs}
    &P_1 \,\theta( U_{02} U_{03}^\dagger U_{01} U_{02}^\dagger U_{03} U_{01}^\dagger,\langle 1\rangle + \langle 2\rangle)\\
    =&-P_1\,\theta ( U_{01},\langle 0\rangle + \langle 2\rangle ) 
    + P_1\,\theta ( U_{03},\langle 0\rangle + \langle 2\rangle) 
    - P_1\,\theta ( U_{02},\langle 0\rangle + \langle 3\rangle )\\
    &+ P_1\,\theta ( U_{01},\langle 0\rangle + \langle 3\rangle )
    - P_1\,\theta ( U_{03},\langle 0\rangle + \langle 1\rangle)
    + P_1\,\theta ( U_{02},\langle 0\rangle + \langle 1\rangle )\\
    =&- \theta( U_{01},\varnothing) + \theta( U_{01},\varnothing) = 0~,
\end{eqs}
where $\varnothing$ denotes the vacuum configuration. The same computation applies to all other vertices, confirming that the $T$-junction process indeed satisfies the local cancellation criterion and therefore forms a statistical process.

We now turn to the 56-step process $\mu_{56}$. As in the $T$-junction computation above, it is necessary to verify that, after projection $P_v$, the total phase vanishes for every vertex $v$. This condition can be refined further: for each choice of $v$, all contributions associated with a given operator $U(s)$ must cancel. Verifying this for all possible pairs $(v,s)$ ensures that $\mu_{56}$ satisfies the local cancellation criterion.

As an illustration, consider $v=\langle 5\rangle$. After applying the projection $P_5$, we are going to show all terms involving $U_{0345}$ cancel. From the intermediate steps listed in SM~Sec.~\ref{app:Intermediate states in the 56-step process}, the relevant contributions are  
\begin{enumerate}
    \item $\theta\!\left(U_{0345}, \langle 015\rangle - \langle 035\rangle + \langle 135\rangle\right)$,  
    \item $\theta\!\left(U_{0345}, \langle 025\rangle - \langle 035\rangle + \langle 235\rangle\right)$,  
    \item $\theta\!\left(U_{0345}, \langle 015\rangle - \langle 035\rangle + \langle 125\rangle + \langle 235\rangle\right)$,  
    \item $\theta\!\left(U_{0345}, \langle 015\rangle - \langle 035\rangle + \langle 145\rangle + \langle 235\rangle - \langle 245\rangle \right)$,  
    \item $\theta\!\left(U_{0345}^\dagger, \langle 015\rangle - \langle 045\rangle + \langle 135\rangle + \langle 345\rangle\right)$,  
    \item $\theta\!\left(U_{0345}^\dagger, \langle 025\rangle - \langle 045\rangle + \langle 235\rangle + \langle 345\rangle\right)$,  
    \item $\theta\!\left(U_{0345}^\dagger, \langle 015\rangle - \langle 045\rangle + \langle 125\rangle + \langle 235\rangle + \langle 345\rangle\right)$,  
    \item $\theta\!\left(U_{0345}^\dagger, \langle 015\rangle - \langle 045\rangle + \langle 145\rangle + \langle 235\rangle - \langle 245\rangle + \langle 345\rangle\right)$.  
\end{enumerate}
From Eq.~\eqref{eq:theta_dagger_def}, we have  
\begin{equation}
    P_5~\theta\!\left(U_{0345}^\dagger, \langle 015\rangle - \langle 045\rangle + \langle 135\rangle + \langle 345\rangle\right) 
    = - P_5~\theta\!\left(U_{0345}, \langle 015\rangle - \langle 035\rangle + \langle 135\rangle\right)~.
\end{equation}
Hence, phases~1 and~5 cancel, as do phases~2 and~6, phases~3 and~7, and phases~4 and~8. Thus, under $P_5$, all contributions cancel exactly. Repeating this argument for every vertex shows that, in the 56-step process, all projected local phase contributions cancel. Therefore, the 56-step unitary satisfies the local cancellation criterion and defines a statistical process for membranes.

\section{Membrane statistics in $d\ge 4$ spatial dimensions}\label{app: Membrane statistics in d>=4 spatial dimensions}

In this Appendix, we study membrane statistics in spatial dimensions $d\ge 4$. We construct higher-form SPT phases in $(d{+}1)$ spatial dimensions and show that symmetry domain walls on their boundaries support membrane excitations with nontrivial statistics, $\mu_{56}\neq 1$.

\subsection{4 spatial dimensions}

First, consider the $(5{+}1)$D 1-form $\mathbb{Z}_N$ SPT phase described by the cocycle
\begin{equation}
    T[B_2] := \frac{1}{N} B_2 \cup B_2 \cup B_2~.
\end{equation}
On a spatial 5-manifold $M_5$, the corresponding SPT ground state can be written as
\begin{equation}
\begin{split}
\ket{\Psi_{\mathrm{SPT}}} &= \sum_{b_1}\exp\!\left(2\pi i\int_{M_5}\Phi[b_1]\right)\ket{b_1}~,\\
\Phi[b_1] &:= \frac{1}{N}\,b_1\cup \delta b_1\cup \delta b_1~.
\end{split}
\end{equation}
where $b_1 \in C^1(M_5, \ZZ_N)$ is a $1$-cochain representing the physical degrees of freedom.  
On a closed manifold $M_5$, this ground state is invariant under the $1$-form symmetry
\begin{equation}
    b_1 \;\to\; b_1 + \delta \epsilon_0, \quad \forall\, \epsilon_0 \in C^0(M_5, \ZZ_N)~.
\end{equation}
If $M_5$ has a boundary, the symmetry action acquires an additional boundary contribution:
\begin{eqs}
    & \frac{1}{N} \int_{M_5} \Big[(b_1+\delta \epsilon_0) \cup \delta (b_1+\delta \epsilon_0) \cup \delta (b_1+\delta \epsilon_0) - b_1 \cup \delta b_1 \cup \delta b_1\Big]  \\
    =& 
    \frac{1}{N} \int_{\partial M_5} \epsilon_0 \cup \delta b_1 \cup \delta b_1 
    =: \int_{\partial M_5} F[b_1, \epsilon_0]~.
\end{eqs}
Thus, on the $(4{+}1)$D boundary the anomalous $1$-form symmetry is
\begin{equation}
    S_{\epsilon_0} := X_{\delta \epsilon_0} \exp \!\left( 2 \pi i \int_{\partial M_5} F[b_1, \epsilon_0] \right)~,
\end{equation}
where $X_{\lambda_1}$ is the Pauli $X$ operator, $\ket{b_1} \to \ket{b_1 + \lambda_1}$.
To enumerate the symmetric operators, we first define the flux (domain-wall) operators
\begin{equation}
    W_f := Z_{\partial f},
\end{equation}
i.e., the product of $Z_e^{\pm1}$ along $\partial f$, with the sign set by orientation.
We then define “hopping terms,” which commute with $S_{\epsilon_0}$ but violate the flux $W_f$:
\begin{eqs}
    & U_{h_1} = X_{h_1} \exp \!\left( 2\pi i \int_{\partial M_5} L[b_1, h_1]\right)~, \\ 
    & L[b_1, h_1] := -\frac{1}{N}\, b_1 \cup (h_1 \cup \delta b_1 + \delta b_1 \cup h_1 + h_1 \cup \delta h_1)~,
\label{eq:hopping operator of membrane in 4d}
\end{eqs}
with $h_1$ any $1$-cochain.
Choose $h_1$ to be $1$ on a single edge $e$ and $0$ elsewhere. The corresponding operator $U_e$ violates the flux terms on all faces $f$ with $e \subset \partial f$. Since the boundary manifold is 4-dimensional, these violations form domain-wall membrane excitations. We verify numerically that the hopping operator defined in Eq.~\eqref{eq:hopping operator of membrane in 4d} generates the statistics $\mu_{56} = \exp(2\pi i/N)$.

Similarly, for the cocycle corresponding to the higher Pontryagin power in Eq.~\eqref{eq:higher Pontryagin power},
\begin{equation}
    T[B_2]:=\frac{1}{3N}
    \big[
     B_2 \cup B_2 \cup B_2 
    + (B_2 \cup_1 \delta B_2)\cup B_2 
    - B_2 \cup (B_2 \cup_1 \delta B_2)
    \big],
\end{equation}
the resulting membrane statistics are given by $\mu_{56} = \exp(2\pi i/3N)$ when $N \equiv 0 \pmod{3}$.

\subsection{5 spatial dimensions}

We now turn to anyonic membrane statistics in $d=5$, where membrane excitations are Poincar\'e dual to a $3$-cochain $B_3$. Accordingly, we consider $(6{+}1)$D SPT phases protected by a $2$-form $\ZZ_N$ symmetry. The corresponding cohomology group has up to two nontrivial cyclic factors. The first contribution is the {\color{blue}$\ZZ_N$} class generated by

\begin{equation}
    {\color{blue} \frac{1}{N^2} B_3 \cup \delta B_3 } \in H^7(B^3\ZZ_N, \RR/\ZZ)~,
\end{equation}
which is the $(6{+}1)$D analogue of a Chern–Simons term and has been extensively studied in M theory in the context of 4-form fluxes~\cite{Witten:1996hc, Witten:1998wy}.
It is straightforward to see that this anomaly cannot be detected by the 56-step process. The cocycle consists of a product of two $B_3$’s, corresponding to a Clifford circuit (after extending the group from $\ZZ_N$ to $\ZZ_{N^2}$). The resulting theory can be formulated in terms of Pauli stabilizer models~\cite{Ellison2022Pauli}, where all excitation operators, including the volume operator $U_{ijkl}$, are Pauli matrices of $\ZZ_{N^2}$. By the following lemma, $\mu_{56}$ must therefore be trivial:  
\begin{lemma}
If all volume operators $U_t$ are Pauli operators, so that 
$$
    [U_{t_1}, [U_{t_2}, U_{t_3}]] = 1~, \quad \text{for all } t_1,t_2,t_3
$$
with $[A,B]:=A^{-1}B^{-1}AB$, then $\mu_{56}$ defined in Eq.~\eqref{eq:56-step_unitary_process} equals $+1$.
\label{lemma: mu56 = 1 for Pauli}
\end{lemma}
\begin{proof}
    If each $U_t$ is a Pauli operator, exchanging $U_{t_1}$ and $U_{t_2}$ produces only a phase depending on $t_1$ and $t_2$.  
    It is then straightforward to verify that all $U_t$ and $U_t^\dagger$ in Eq.~\eqref{eq:56-step_unitary_process} cancel pairwise, and the accumulated phases also vanish.  
\end{proof}
Therefore, we focus on the other contribution, the {\color{red}$\ZZ_{\gcd(3, N)}$} class generated by   
\begin{equation}
    {\color{red} \frac{1}{3} P^1(B_3) := \frac{1}{3}D^3_2(B_3)} 
    \;\in\; H^7(B^3\ZZ_N, \RR/\ZZ)~, \quad \text{for } N\equiv 0 \pmod{3}~,
\label{eq:P1(B3)}
\end{equation}
where $P^i$ denotes the Steenrod reduced $p$-th power operation (with $p=3$ in our case):  
\begin{equation}
    P^i : H^n(X;\, \ZZ_p) \;\to\; H^{\,n+2i(p-1)}(X;\, \ZZ_p)~,
\end{equation}
and $D^p_i$ denotes the Steenrod cup-$(p,i)$ product introduced in Ref.~\cite{Medina2021CochainMaySteenrod}, which realizes the cochain-level May–Steenrod operations~\cite{Steenrod1947, May1970algebraic}. The explicit cochain expression of the cup-$(p,i)$ product is reviewed in SM~Sec.~\ref{app:Higher cup products, May–Steenrod structure, and Steenrod reduce powers}.  
Moreover, on any closed oriented $(n+4)$-dimensional manifold $M$, for $B_n\in H^n(M,\ZZ_3)$, $P^1(B_n)$ is shown in Ref.~\cite{Tomonaga1965PontryaginMod3} to be related to the first Pontryagin class $p_1=p_1(TM)\in H^4(M,\ZZ)$ through the top-degree cohomology relation:  
\begin{equation}
    P^1(B_n) = p_1 \cup B_n \pmod{3}~.
\end{equation}
We therefore refer to this ${\ZZ_{\gcd(3, N)}}$ sector as \textbf{Pontryagin statistics}. By contrast, the familiar fermionic $\ZZ_2$ statistics arises from Wu classes, which can be rephrased in terms of Stiefel–Whitney classes.

Given this cochain-level representative of the cohomology class, one can follow the standard procedure~\cite{chen2013groupcohomology, Kapustin2017Higher} to construct the corresponding wavefunction and parent Hamiltonian of the SPT phase.
Using the cochain expressions in SM~Sec.~\ref{app:Higher cup products, May–Steenrod structure, and Steenrod reduce powers} and the algorithm in SM~Sec.~\ref{sec:Algorithm for verifying membrane statistics}, we construct the $(6{+}1)$D $2$-form $\ZZ_N$ SPT phase and show that domain-wall membrane excitations on its $(5{+}1)$D boundary exhibit nontrivial statistics $\mu_{56} = \exp(\frac{2\pi i}{3})$ for $N\equiv 0 \pmod{3}$.

\subsection{6 spatial dimensions}

Now we turn to 6 spatial dimensions. In this case, $H^8(B^4\ZZ_N, \RR/\ZZ)$ has up to two nontrivial cyclic factors. The first is the {\color{blue}$\ZZ_{N \times \gcd(2,N)}$} class, represented by  
\begin{equation}
    {\color{blue}
    \begin{cases}
        \frac{1}{N}\, B_4 \cup B_4~, & N \not\equiv 0 \pmod{2}~,\\[3pt]
        \frac{1}{2N}\!\left(B_4 \cup B_4 + B_4 \cup_1 \delta B_4 \right)~, & N \equiv 0 \pmod{2}~.
    \end{cases}}
\end{equation}
Since this cocycle involves only products of two $B_4$, it can be realized by a Clifford circuit after a suitable group extension~\cite{Ellison2022Pauli}. Consequently, these phases admit a Pauli stabilizer realization, and the 56-step process is trivial, $\mu_{56}=1$, unable to detect the corresponding anomalies.

The other generator, {\color{red}$\ZZ_{\gcd(3,N)}$}, is given by  
\begin{equation}
    {\color{red} \frac{1}{3} P^1(B_4) := -\frac{1}{3} D^3_4(B_4) }
    \;\in\; H^8(B^4\ZZ_N, \RR/\ZZ)~, \quad \text{for } N\equiv 0 \pmod{3}~,
\label{eq:P1(B4)}
\end{equation}
from which one can construct the corresponding wave function and parent Hamiltonian of the SPT phase.  
Since $D^3_4(B_4)$ expands into 177 terms, however, the explicit wave function is cumbersome to write down by hand.  
To address this, we developed an algorithm that extracts the 56-step invariant $\mu_{56}$ directly from the cocycle $\tfrac{1}{3} P^1(B_4)$, as detailed in SM~Sec.~\ref{sec:Algorithm for verifying membrane statistics}.
Using this approach, we construct the $(7{+}1)$D $3$-form $\ZZ_N$ SPT phase and show that domain-wall membrane excitations on its $(6{+}1)$D boundary exhibit nontrivial statistics $\mu_{56} = \exp(\frac{2\pi i}{3})$ for $N\equiv 0 \pmod{3}$.

\subsection{7 spatial dimensions}

Finally, we consider 7 spatial dimensions. For all higher dimensions, the results coincide with the 7-dimensional case since the cohomology stabilizes~\cite{LodayCyclicHomology}.
First, the {\color{blue}$\ZZ_{\gcd(2,N)}$} class is generated by  
\begin{equation}
    \frac{1}{2} Sq^4 B_5 \;=\; \frac{1}{2} B_5 \cup_1 B_5 \;\in\; H^{9}(B^{5}\ZZ_N, \RR/\ZZ)~, \quad \text{for } N\equiv 0 \pmod{2}~.
\end{equation}
As this corresponds to a Pauli stabilizer model, the 56-step process is trivial, $\mu_{56}=1$, and cannot detect the associated statistics.  

In contrast, the {\color{red}$\ZZ_{\gcd(3,N)}$} class is generated by  
\begin{equation}
    {\color{red} \frac{1}{3} P^1(B_5) :=  -\frac{1}{3} D^3_6(B_5) }
    \;\in\; H^9(B^5\ZZ_N, \RR/\ZZ)~, \quad \text{for } N\equiv 0 \pmod{3}~,
\label{eq:P1(B5)}
\end{equation}
where the cochain-level expression of $D^3_6(B_5)$ is given in SM~Sec.~\ref{app:Higher cup products, May–Steenrod structure, and Steenrod reduce powers} and involves 1110 terms built from triple products of $B_5$.  
Our computational algorithm shows that the 56-step process on the $(7{+}1)$D boundary of the corresponding $(8{+}1)$D SPT phase yields $\mu_{56} = \exp(\frac{2\pi i}{3})$ for $N\equiv 0 \pmod{3}$.

\section{Computational derivation of $\mu_{56}$}\label{app:Computational derivation}

As a statistical process for $\ZZ$-membranes in four spatial dimensions, $\mu_{56}$ appears unusual and highly asymmetric. It was not constructed by hand; instead, it was obtained by a computer-assisted search. While the algorithms of Refs.~\cite{xue2025statistics,kobayashi2024generalized} are sufficient to find statistical processes for $\ZZ_2$-membranes, extending them to $\ZZ$-membranes is substantially more technical and is not guaranteed to succeed.

Refs.~\cite{kobayashi2024generalized,xue2025statistics} reformulate the search for statistical processes as a linear-algebra problem and solve it using Smith normal form techniques. Due to computational complexity, those methods are practical mainly for $\ZZ_2$-membranes. A direct extension to $G=\ZZ$ is infeasible because the relevant spaces become infinite-dimensional. Our strategy is therefore to first obtain a process for $\ZZ_2$-membranes and then attempt to lift it to a process for $\ZZ$-membranes.

We briefly recall the framework of Ref.~\cite{xue2025statistics}. The basic input, the fusion group $G$, excitation dimension $p$, and spatial dimension $d$, is packaged into an \textbf{excitation pattern}.
\begin{definition}\label{defexcitationModel}
An \textbf{excitation pattern} $m$ in $d$ spatial dimensions consists of data $(A,S,\partial,\supp)$:
\begin{enumerate}
    \item An Abelian group $A$, called the \emph{configuration group}.
    \item A finite set $S$ and a map $\partial:S\to A$ such that $\{\partial s\mid s\in S\}$ generates $A$.
    \item For each $s\in S$, a region $\supp(s)\subset S^d$, where $S^d$ denotes the $d$-sphere.
\end{enumerate}
\end{definition}

In this work we specialize to the following simplicial models.
\begin{definition}
Given a cyclic fusion group $G=\ZZ$ or $\ZZ_N$ with a chosen unit generator, excitation dimension $p$, and spatial dimension $d$, we take
\begin{enumerate}
    \item $A=B_p(\partial\Delta^{d+1},G)$, the Abelian group of $p$-boundaries on $\partial\Delta^{d+1}\simeq S^d$, so we consider only closed excitations.
    \item $S$ to be the set of $(p{+}1)$-simplices of $\Delta^{d+1}$, with $\supp(s)$ equal to the simplex itself.
    \item $\partial$ to be the usual chain boundary map, sending a $(p{+}1)$-simplex $s$ to $\partial s\in B_p(\partial\Delta^{d+1},G)$.
\end{enumerate}
A \textbf{process} (such as $\mu_{56}$) is a word in generators $s\in S$ and their inverses, i.e.\ an element of the free group $\operatorname{F}(S)$.
\end{definition}

For fixed $(G,p,d)$, many lattice systems can realize the same  excitation pattern. Although their anomalies may differ (classified by $H^{d+2}(B^{d-p}G,\RR/\ZZ)$), they share universal locality constraints captured by the notion of realization.
\begin{definition}\label{defRealization}
A \textbf{realization} of an excitation pattern $m$ consists of a Hilbert space $\mathcal{H}$, a collection of orthonormal configuration states $\{\ket{a}|a\in A\}$, and a collection of \textit{excitation operators} $\{U(s)|s\in S\}$, satisfying the following two axioms.
	\begin{itemize}
			\item \textbf{The configuration axiom:} For any $s\in S$ and $a\in A$, the equation
			\begin{equation}\label{eqChangeConfig}
				U(s)|a\rangle=e^{i\theta(s,a)}|a+\partial s\rangle
			\end{equation}
			holds for some $\theta(s,a)\in \RR/2\pi\ZZ$.
			\item \textbf{The locality axiom:} $\forall s_1, s_2, \dots, s_k \in S$ satisfying $\operatorname{supp}(s_1) \cap \operatorname{supp}(s_2) \cap \cdots \cap \operatorname{supp}(s_k) = \emptyset$, the equation
			\begin{align}\label{axiomLocalityIdentity}
				[ U(s_k), [\cdots, [ U(s_2),  U(s_1)]]]=1 \in U(\mathcal{H}),
			\end{align}
			holds, where $[a,b] = a^{-1}b^{-1}ab$.
	\end{itemize}
The map $U:S\to U(\mathcal{H})$ extends to a homomorphism $\operatorname{F}(S)\to U(\mathcal{H})$, and Eq.~\eqref{axiomLocalityIdentity} is understood in this sense, i.e.\ $[U(s_k),[\cdots,[U(s_2),U(s_1)]]]=1$.
\end{definition}
The locality identity is automatic when supports are disjoint: if $\supp(s_1)\cap\supp(s_2)=\emptyset$, then $U(s_1)$ and $U(s_2)$ act on disjoint degrees of freedom, so $U([s_1,s_2])=1$. More generally, for finite-depth circuits $U(s)$, the commutator $U([s_2,s_1])$ is supported near $\supp(s_1)\cap\supp(s_2)$, hence $U([s_3,[s_2,s_1]])=1$ whenever $\supp(s_1)\cap\supp(s_2)\cap\supp(s_3)=\emptyset$, and similarly for higher nested commutators.

As in the main text, we write
$$
U(s)\ket{a}=e^{i\theta(U(s),a)}\ket{a+\partial s}.
$$
For $g\in\operatorname{F}(S)$, define $U(g)\ket{a}=e^{i\theta(U(g),a)}\ket{a+\partial g}$. Using
\begin{equation}\label{eqThetaP(g,a)}
\begin{aligned}
\theta(U(g_2g_1),a)&=\theta(U(g_1),a)+\theta(U(g_2),a+\partial g_1),\\
\theta(U(g^{-1}),a)&=-\theta(U(g),a-\partial g),
\end{aligned}
\end{equation}
we can expand $\theta(U(g),a)$ uniquely as an integer linear combination of phases $\theta(U(s),a)$. We regard the physically relevant data of a realization as the collection of diagonal matrix elements $\langle a|U(g)|a\rangle$ with $\partial g=0$; accordingly, two realizations are identified if their phase data $\{\theta(U(s),a)\in\RR/(2\pi\ZZ)\mid s\in S,\ a\in A\}$ agree.

To formalize measurement outcomes, let $\theta(g,a)$ denote the linear functional that maps a realization to $\theta(U(g),a)\in\RR/(2\pi\ZZ)$. It decomposes into the generators $\theta(s,a)$ via the relations
\begin{equation}\label{eqTheta(g,a)}
\begin{aligned}
\theta(g_2g_1,a)&=\theta(g_1,a)+\theta(g_2,a+\partial g_1),\\
\theta(g^{-1},a)&=-\theta(g,a-\partial g).
\end{aligned}
\end{equation}
Let $E$ be the free Abelian group generated by the symbols $\theta(s,a)$, representing all such measurements. An element of $E$ is an \textbf{expression} of the form $\sum_i c_i\,\theta(s_i,a_i)$ with $c_i\in\ZZ$. The subset $\{\theta(g,a)\mid g\ \text{is a statistical process}\}$ forms a subgroup $E_{\mathrm{inv}}\subset E$, whose elements are \textbf{statistical expressions}. The search for statistical processes can thus be reduced to linear algebra.

Locality imposes further structures. From Eq.~\eqref{axiomLocalityIdentity}, for any $s_1,\ldots,s_k$ with empty common intersection of supports, the statistical expression $\theta([s_k,[\cdots,[s_2,s_1]]],a)\in E_{\mathrm{inv}}$ evaluates to zero in every realization. Such relations generate a subgroup $E_{\mathrm{id}}\subset E_{\mathrm{inv}}$ of expressions that are physically trivial. Consequently, statistical expressions are classified by
$$
T:=E_{\mathrm{inv}}/E_{\mathrm{id}},
$$
which is Pontryagin dual to the anomaly classification by construction. When the spatial simplicial complex is $\partial\Delta^{d+1}$, Theorem III.2 of Ref.~\cite{xue2025statistics} gives $T\simeq H_{d+2}(B^{d-p}G,\ZZ)$, dual to $H^{d+2}(B^{d-p}G,\RR/\ZZ)$. Although $E_{\mathrm{inv}}$ and $E_{\mathrm{id}}$ are infinite-dimensional for $G=\ZZ$, the quotient $T$ is still well defined and finitely generated.

In practice, we compute $T$ by extracting the torsion subgroup of $E/E_{\mathrm{id}}$, which agrees with $T$ when $G=\ZZ_n$ (see Theorem~VI.9 of Ref.~\cite{xue2025statistics}). Concretely, we assemble a sparse integer matrix whose rows encode generators of $E_{\mathrm{id}}$ (most nonzero entries are $\pm1$), and perform integer row operations to eliminate $\pm1$ entries. After sufficient reductions, the residual matrix typically has very small rank. For $\ZZ_2$-membranes in $d=4$, this procedure yields a residual rank-one relation generated by an expression $e$ divisible by $2$, from which one concludes $T=\ZZ_2$ generated by $e/2$. Writing $e/2=\theta(g,0)$ then produces a nontrivial statistical process $g$ for $\ZZ_2$-membranes.

The resulting $g$ may not define a valid process for $\ZZ$-membranes because its local cancellation condition is stricter. We write
$$
    \theta(g,0)=\sum_{s\in S,\ a\in B_2(\Delta^5,\ZZ)} c(s,a)\,\theta(s,a),\qquad c(s,a)\in\ZZ.
$$
The $\ZZ$ cancellation condition requires that, for every fixed $s\in S$, vertex $x\in\supp(s)$, and $\alpha\in C_2(\Delta^5,\ZZ)|_x$ (the free Abelian group generated by $2$-simplices containing $x$), one has
\begin{equation}\label{eqLocalCancellation}
    \sum_{a:\ a|_x=\alpha} c(s,a)=0.
\end{equation}
By contrast, for $\ZZ_2$-membranes the constraint only depends on the class $[\alpha]\in C_2(\Delta^5,\ZZ_2)|_x$:
$$
\sum_{a:\ [a|_x]=[\alpha]} c(s,a)=0.
$$
Since distinct $\alpha$ can map to the same $[\alpha]$, the $\ZZ$ condition imposes many more independent constraints.

To bridge this gap, we impose an auxiliary restriction. We choose a function $f$ that assigns $\pm1$ to each $2$-simplex of $\Delta^5$ and restrict to processes $g$ whose intermediate configurations $a\in B_2(\Delta^5,\ZZ)$ (along the action of $g$ on the vacuum) satisfy
\begin{equation}
    a(t)\in\{0,f(t)\}\quad \text{for every $2$-simplex $t$}.
\end{equation}

Under this restriction, Eq.~\eqref{eqLocalCancellation} becomes effectively a single condition for each $(s,[\alpha])$ at fixed $x$, so the $\ZZ$ and $\ZZ_2$ cancellation criteria align.

This condition can be implemented at the level of expressions. Given $a\in B_2(\Delta^5,\ZZ_2)$, define its lift $\tilde a\in C_2(\Delta^5,\ZZ)$ by $\tilde a(t)=f(t)a(t)$, and call $a$ \textbf{legal} if $\tilde a\in B_2(\Delta^5,\ZZ)$. We further call $\theta(s,a)$ \textbf{legal} if $a$ is legal and, for every $2$-simplex $t$,
\begin{equation}\label{eqLegal}
(\partial s+\tilde a)(t)\in\{0,f(t)\},
\end{equation}
where $\partial s$ is taken with $\ZZ$ coefficients. If every term in a statistical expression is legal, we obtain a corresponding $\ZZ$-membrane expression by replacing $a$ with $\tilde a$. In the elimination procedure above, we therefore only use row operations that eliminate $\pm1$ entries in \emph{illegal} columns, and we stop once no such entries remain. If all illegal columns are eliminated, the remaining relations involve only legal terms and may yield a nontrivial legal statistical expression; otherwise, the attempt fails.

Because legality depends on the choice of $f$, many trials are required. Most choices fail, and each run takes a few hours on a standard desktop computer, making a naive search impractical. We accelerate the search by exploiting shared structure among different choices of $f$: if $f_1(t)=f_2(t)$ for a simplex $t$, then the legality condition \eqref{eqLegal} coincides on $t$, allowing common elimination steps. Using this optimization, we ran the search in parallel on 20 CPUs and found a suitable $f$ after several hours.
Continuing the reduction yields the process $\mu_{56}$. Notably, all operators appearing in $\mu_{56}$ have supports containing the vertex $0$.

\section{Higher cup products, May–Steenrod structure, and Steenrod reduced powers}\label{app:Higher cup products, May–Steenrod structure, and Steenrod reduce powers}

We begin with a short review of simplicial cohomology, following the notation of Refs.~\cite{Chen2021Disentangling, Chen2023Highercup}. We recall the definitions of cochains, the coboundary operator, and cup products, together with their higher versions. For generalizations to group and higher-group cohomology, see Refs.~\cite{chen2012symmetry, Kapustin2017Higher}.
A related survey of these topics can be found in Ref.~\cite{M5}.

\subsection{Simplicial Cohomology} \label{app:terminology}

Let $M$ be a simplicial complex or a cellulation of a manifold. An $n$-simplex (or cell) is denoted by $\sigma_n$. In low dimensions we use the usual names: $\sigma_0 \equiv v$ for a vertex, $\sigma_1 = e$ for an edge, $\sigma_2 = f$ for a face, and $\sigma_3 = t$ for a tetrahedron or $\sigma_3 = c$ for a cube. In general, a $p$-simplex is specified by its vertices, written as $\langle 0, 1, 2,\ldots, p \rangle$. A $p$-chain is a formal linear combination of $p$-simplices.
The collection of $p$-chains forms the group $C_p(M,\ZZ)$. Writing a generic $p$-chain as $c_p$, we have $c_p = \sum_{\sigma_p} a_{\sigma_p} \sigma_p$ with coefficients $a_{\sigma_p} \in \ZZ$.

The boundary operator 
\begin{equation}
    \partial: C_p(M, \ZZ) \to C_{p-1}(M, \ZZ)
\end{equation}
maps a $p$-simplex to the sum of its oriented boundary simplices. For the oriented $p$-simplex $\langle 0, 1, \dots, p \rangle$, the boundary is 
\begin{equation}
    \partial \langle 0, 1, \dots, p \rangle = \sum_{i=0}^p ~(-1)^i ~\langle 0,\dots, \hat{i},\dots,p \rangle,
\end{equation}
where $\langle 0,\dots, \hat{i},\dots,p \rangle$ denotes the $i$th face obtained by removing vertex $i$.

Let $R$ be a ring, typically assumed to be $R=\ZZ$ or $\ZZ_N$. A $p$-cochain is a function that assigns an $R$-value to each $p$-simplex, linearly extending to a map from $C_p(M,\ZZ)$ to $R$. The set of all $p$-cochains is denoted $C^p(M, R)$. For each vertex $v$, we define a ``dual'' 0-cochain $\boldsymbol{v}$ via
\begin{align}
    \boldsymbol{v}(v') = 
    \begin{cases}
      1\in R & \text{if } v' = v,\\
      0 & \text{otherwise}.
    \end{cases}
    \label{example: chain cochain correspondance}
\end{align}
Here we use bold symbols to distinguish cochains from chains. Similarly, we define $\boldsymbol{e}$ and $\boldsymbol{f}$ for each edge $e$ and face $f$. For $c_p= \sum_{\sigma_p} a_{\sigma_p} \sigma_p\in C_p(M,\ZZ)$, we define $\boldsymbol{c}_p = \sum_{\sigma_p} a_{\sigma_p} \boldsymbol{\sigma}_p\in C^p(M,R)$.

The coboundary operator $\delta$ is dual to $\partial$, mapping a $p$-cochain to a $(p+1)$-cochain:
\begin{align}
    \delta: C^p(M, R) \to C^{p+1}(M,R).
\end{align}
For a $p$-cochain $\boldsymbol{c}$ and a $(p+1)$-simplex $s$, it is defined using $\partial$ as
\begin{align} \label{coboundarydef1}
    \delta \boldsymbol{c}(s) = \boldsymbol{c}(\partial s).
\end{align}
The cup product $\cup$ combines a $p$-cochain and a $q$-cochain into a $(p+q)$-cochain:
\begin{align}
    \cup: C^p(M, R) \times C^q(M, R) \to C^{p+q}(M, R).
\end{align}
For example, if $\boldsymbol{c} \in C^p(M, \mathbb{Z}_2)$ and $\boldsymbol{d} \in C^q(M, \mathbb{Z}_2)$, then on a $(p+q)$-simplex $\langle 0, 1, \ldots, p+q \rangle$,
\begin{align}
    \bigl(\boldsymbol{c} \cup \boldsymbol{d}\bigr)\bigl(\langle 0, 1,\ldots, p+q \rangle\bigr) 
    = \boldsymbol{c}\bigl(\langle 0, 1,\ldots, p \rangle\bigr)\,\boldsymbol{d}\bigl(\langle p, p+1, \ldots, p+q \rangle\bigr).
\end{align}
The coboundary satisfies the Leibniz rule:
\begin{align} \label{leibnizrule}
    \delta (\boldsymbol{c} \cup \boldsymbol{d})
    = \delta \boldsymbol{c} \cup \boldsymbol{d} + (-1)^p\boldsymbol{c} \cup \delta \boldsymbol{d},
\end{align}
for any $p$-cochain $\boldsymbol{c}$ and $q$-cochain $\boldsymbol{d}$.

One can also define higher cup products. A cup-$i$ product $\cup_i$ maps a $p$-cochain and a $q$-cochain to a $(p+q-i)$-cochain:
\begin{align}
    \cup_i : C^p(M, R) \times C^q(M, R) \to C^{p+q-i}(M,R).
\end{align}
Explicit formulas for higher cup products on simplicial complexes are given in Refs.~\cite{Steenrod1947, Gaiotto2015SpinTQFTs, lokman2020lattice, M1, M2}. They obey the recursive relation
\begin{align} \label{highercupleibniz}
    \delta (\boldsymbol{c} \cup_i \boldsymbol{d}) 
    = \delta \boldsymbol{c} \cup_i \boldsymbol{d} 
    + (-1)^p\boldsymbol{c} \cup_i \delta \boldsymbol{d} 
    + (-1)^{p+q-i}\boldsymbol{c} \cup_{i-1} \boldsymbol{d} 
    + (-1)^{pq+p+q}\boldsymbol{d} \cup_{i-1} \boldsymbol{c},
\end{align}
with $\cup_0 \equiv \cup$.

For the cup–1 product we have
\begin{eqs}
    \bigl(\boldsymbol{c}_p \cup_1 \boldsymbol{d}_q\bigr)\bigl(\langle 0, \ldots, p+q-1 \rangle\bigr) 
    &= \sum_{i=0}^{p-1} (-1)^{(p-i)(q+1)}\boldsymbol{c}_p\bigl(\langle 0, \ldots, i, q+i, \ldots, p+q-1 \rangle\bigr)~
    \boldsymbol{d}_q\bigl(\langle i, \ldots, q+i \rangle\bigr) \\
    &\equiv \sum_{i=0}^{p-1} (-1)^{(p-i)(q+1)} \boldsymbol{c}_p\bigl(\langle 0 \rightarrow i,~ q+i \rightarrow p+q-1 \rangle\bigr)~
    \boldsymbol{d}_q\bigl(\langle i \rightarrow q+i \rangle\bigr).
\end{eqs}
For the cup–2 product (non-zero only when $2\le p,q$), we obtain
\begin{eqs}
    \bigl(\boldsymbol{c}_p \cup_2 \boldsymbol{d}_q\bigr)\bigl(\langle 0, \ldots, p+q-2 \rangle\bigr)
    &=\sum_{\substack{0\le i_1\le p-2\\ i_1<i_2\le q+i_1-1}}
    (-1)^{(p-i_1)(i_2-i_1-1)}~
    \boldsymbol{c}_p (\langle 0\rightarrow i_1, ~i_2 \rightarrow p+i_2-i_1-1 \rangle)\\
    & \qquad \qquad \boldsymbol{d}_q (\langle i_1\rightarrow i_2, ~p+i_2-i_1-1 \rightarrow p+q-2\rangle).
\end{eqs}
For a general integer $k$ with $0\le k\le\min\{p,q\}$, put $m=p+q-k$.
Choose $k+1$ cut vertices $0\le a_0<\cdots<a_k\le m$, and set
$a_{-1}=0$ and $a_{k+1}=m$. For $j=0,\ldots,k+1$, let
$I_j=(a_{j-1},a_{j-1}+1,\ldots,a_j)$ be the ordered list of vertices
in the corresponding interval. Let $A$ be the concatenation of the
$I_j$ with even $j$, and $B$ the concatenation of the $I_j$ with odd $j$.
We retain only the choices of cuts for which $|A|=p+1$ and $|B|=q+1$,
where $|A|$ and $|B|$ denote the numbers of vertices in these lists.
Then the cup-$k$ product is
\begin{equation}
\begin{aligned}
    &\bigl(\boldsymbol{c}_p\cup_k\boldsymbol{d}_q\bigr)
      \bigl(\langle0,1,\ldots,m\rangle\bigr)\\
    &\qquad=
    \sum_{\substack{0\le a_0<\cdots<a_k\le m\\
                     |A|=p+1,\ |B|=q+1}}
    (-1)^{\operatorname{inv}(A,B\setminus A)}
    \boldsymbol{c}_p\bigl(\langle A\rangle\bigr)
    \boldsymbol{d}_q\bigl(\langle B\rangle\bigr).
\end{aligned}
\end{equation}
Here $B\setminus A$ means the ordered list obtained from $B$ by deleting
the vertices already present in $A$, and
\begin{equation}
    \operatorname{inv}(A,B\setminus A)
    =\#\bigl\{(u,v)\in A\times(B\setminus A):u>v\bigr\}.
\end{equation}
Thus the sign is the parity of the permutation that rearranges $A$
followed by $B\setminus A$ into the natural vertex order $0,1,\ldots,m$.
For $k=0$, the only admissible cut is $a_0=p$, and this expression
reduces to the ordinary cup product.

Finally, as a convenient shorthand, we write
\begin{align}
    \int_{N_p} \boldsymbol{c}_p \equiv \sum_{s_p \in N_p} \boldsymbol{c}_p(s_p),
\end{align}
where $N_p$ denotes a $p$-dimensional manifold. When the context is clear, we will omit the subscript and simply write $\int \boldsymbol{c}$.

\subsection{May-Steenrod Structure}
While higher cup products provide a simple cochain-level representation of operations, including the Steenrod square and the Pontryagin square, it proves difficult (if not impossible) to use them to construct more complicated operations, such as the Steenrod reduced power~\cite{Steenrod1947, George1970Appreciation, May1970algebraic}.
To obtain cochain-level representations of the Steenrod reduced power, we need a higher version of cup products called the May-Steenrod structure~\cite{Medina2021CochainMaySteenrod, M3}.
May-Steenrod structures can be combined with combinatorial operad actions to produce the cochain-level operations that we want. Our exposition largely follows Refs.~\cite{Medina2021CochainMaySteenrod, M4}, in which the process to obtain the Steenrod reduced powers at the cochain level is explained in detail.

Let $k,r\geq 0$, and consider a surjection $u:\{1,\cdots,k\}\to \{1,\cdots,r\}$. It is called non-degenerate if we have $u(i)\neq u(i+1)$ for all $1\leq i<k$. In the following, this kind of $u$ is simply called a ``non-degenerate surjection''. We borrow the name from the mathematics literature and call the set of non-degenerate surjections a surjection operad. 

Each element $u:\{1,\cdots,k\}\to \{1,\cdots,r\}$ of this operad corresponds to an operation that combines $r$ cochains to form a new cochain~\cite{MS, Matthias2021Homotopy}: 
\begin{eqs}
    \phi(u):C^{p_1}(M,R)\times C^{p_2}(M,R)\cdots\times C^{p_r}(M,R)\to C^n(M,R)~.
\end{eqs}
Here $n=\sum_{s=1}^r p_s-(k-r)$. Given $r$ cochains $c_1,c_2,\cdots,c_r$, the map $\phi(u)$ is evaluated on an ordered $n$-simplex by 
\begin{eqs}
    \phi(u)(c_1,\cdots,c_r)(\langle0,\ldots,n\rangle)=\sum_{0=i_0\leq i_1\leq \cdots\leq i_{k-1}\leq i_{k}=n}(-1)^{\mathsf{sgn}} c_1(\langle \mu_1\rangle)c_2(\langle\mu_2\rangle)\cdots c_r(\langle\mu_r\rangle)~,
\end{eqs}
where $\mu_s$ is built from the concatenation of the intervals $[i_{t-1},i_t]$ such that $u(t)=s$, and the sign factor is explained below. Cochain evaluations on vertex lists of the wrong degree or with repeated vertices are understood to vanish. 
We call an interval $I=[i_{t-1},i_t]$ final if $u(t')\neq u(t)$ for all $t'>t$, and call it inner otherwise. The length of the interval $I$ is defined to be $\mathrm{len}_u(I)=i_t-i_{t-1}$ if $I$ is final, and $i_t-i_{t-1}+1$ if $I$ is inner. 

Consider a sorting process that takes $(u(1),u(2),\cdots,u(k))$ to $(1,\cdots,1,2,\cdots,2,\cdots,r,\cdots,r)$, in which only adjacent elements with different values can be exchanged. When we exchange two intervals $I_1$ and $I_2$, there is a sign $(-1)^{\mathrm{len}_u(I_1)\mathrm{len}_u(I_2)}$. We take the product of all these signs to obtain the permutation sign. Also, we obtain the position sign by taking the product of $(-1)^{i_t}$ for each inner interval $[i_{t-1},i_t]$. The total sign $(-1)^\mathsf{sgn}$ is the product of the permutation sign and the position sign. 
As an example, we derive the action of the surjection $u=(1,2,3,1,2)$ on $3$-cochains. The result is a $7$-cochain 
\begin{eqs}
    &\phi(u)(c_1,c_2,c_3)(\langle 0, 1, 2, 3, 4, 5, 6, 7 \rangle)
    \\=
    &+c_1(\langle 0,3,4,5\rangle)~c_2(\langle 0,5,6,7\rangle)~c_3(\langle 0,1,2,3\rangle)
    +c_1(\langle 0,4,5,6\rangle)~c_2(\langle 0,1,6,7\rangle)~c_3(\langle 1,2,3,4\rangle)\\
    &+c_1(\langle 0,5,6,7\rangle)~c_2(\langle 0,1,2,7\rangle)~c_3(\langle 2,3,4,5\rangle)
    +c_1(\langle 0,1,4,5\rangle)~c_2(\langle 1,5,6,7\rangle)~c_3(\langle 1,2,3,4\rangle)\\
    &-c_1(\langle 0,1,5,6\rangle)~c_2(\langle 1,2,6,7\rangle)~c_3(\langle 2,3,4,5\rangle)
    +c_1(\langle 0,1,6,7\rangle)~c_2(\langle 1,2,3,7\rangle)~c_3(\langle 3,4,5,6\rangle)\\
    &+c_1(\langle 0,1,2,5\rangle)~c_2(\langle 2,5,6,7\rangle)~c_3(\langle 2,3,4,5\rangle)
    +c_1(\langle 0,1,2,6\rangle)~c_2(\langle 2,3,6,7\rangle)~c_3(\langle 3,4,5,6\rangle)\\
    &+c_1(\langle 0,1,2,7\rangle)~c_2(\langle 2,3,4,7\rangle)~c_3(\langle 4,5,6,7\rangle)~.\label{Example 12312 on 3cochain}
\end{eqs}
The signs in Eq.~\eqref{Example 12312 on 3cochain} are in accordance with the above sign convention. Note that only the first two intervals are inner, so the length of the term $c_1(\langle 0,3,4,5\rangle)~c_2(\langle 0,5,6,7\rangle)~c_3(\langle 0,1,2,3\rangle)$ is $\mathrm{len}_u=(1,1,3,2,2)$, and the permutation sign is $1$. Also, the position sign is $(-1)^{0+0}=1$. On the other hand, the length of the term $c_1(\langle 0,1,5,6\rangle)~c_2(\langle 1,2,6,7\rangle)~c_3(\langle 2,3,4,5\rangle)$ is $\mathrm{len}_u=(2,2,3,1,1)$, which gives a trivial permutation sign and a position sign $(-1)^{1+2}=-1$. Therefore, we obtain a total sign of $-1$. 

Then we describe the May-Steenrod structure on the surjection operad. For every $r\geq 2,n\geq 0$, we obtain a linear combination of surjections $\psi(r)(e_n)$ following the rules below. 

Let $\rho$ be the cyclic permutation of $\{1,\cdots,r\}$. It acts on each surjection $u:\{1,\cdots,k\}\to \{1,\cdots,r\}$ by $\rho\cdot u(t)=\rho(u(t))$. Using the action $\rho$, we define two operators
\begin{eqs}
    T&=\rho-1\\
    N&=1+\rho+\rho^2+\cdots+\rho^{r-1}
\end{eqs}
that annihilate each other. Furthermore, we define three maps $i,p,s$, satisfying
\begin{align}
    i(u_1,u_2,\cdots,u_k)&=(1,u_1+1,u_2+1,\cdots,u_k+1)\\
    p(u_1,u_2,\cdots,u_{j-1},1,u_{j+1},\cdots,u_k)&=(u_1-1,u_2-1,\cdots,u_{j-1}-1,u_{j+1}-1,\cdots,u_k-1)\\
    s(u_1,u_2,\cdots,u_k)&=(1,u_1,u_2,\cdots,u_k)~.
\end{align}
Degenerate surjections are understood to vanish. When more than one entry equals $1$, we have $p(u_1,u_2,\cdots,u_k)=0$ by convention. The contraction map $h$ is defined as
\begin{eqs}
    h=s+isp+i^2sp^2+\cdots+i^{r-1}sp^{r-1}~.
\end{eqs}
Now, we define $\psi(r)(e_n)$ inductively by
\begin{enumerate}
    \item $\psi(r)(e_0)=(1,2,\cdots,r)$
    \item $\psi(r)(e_{2m+1})=hT\psi(r)(e_{2m})$
    \item $\psi(r)(e_{2m})=hN\psi(r)(e_{2m-1})$
\end{enumerate}
For example, when $r=3$, we obtain $\psi(3)(e_n)$ for the lowest values of $n$, 
\begin{eqs}
    \psi(3)(e_0)=&~(1,2,3)~,\\
    \psi(3)(e_1)=&~(1,2,3,1)~,\\
    \psi(3)(e_2)=&~(1,2,3,1,2)+(1,3,1,2,3)+(1,2,3,2,3)~,\\
    \psi(3)(e_4)=&~(1,2,3,1,2,3,1)+(1,2,3,2,3,1,2)+(1,2,3,1,2,1,2)+(1,3,1,2,3,1,2)\\
    &+(1,3,1,3,1,2,3)+(1,2,3,2,3,2,3)+(1,3,1,2,3,2,3)~,\\
    \psi(3)(e_6)=&~(1,2,3,1,2,3,1,2,3)+(1,2,3,1,2,1,2,3,1)+(1,2,3,1,2,3,1,3,1)\\
    &+(1,2,3,2,3,1,2,3,1)+(1,2,3,2,3,2,3,1,2)+(1,2,3,2,3,1,2,1,2)\\
    &+(1,2,3,1,2,1,2,1,2)+(1,3,1,2,3,1,2,3,1)+(1,3,1,2,3,2,3,1,2)\\
    &+(1,3,1,2,3,1,2,1,2)+(1,3,1,3,1,2,3,1,2)+(1,3,1,3,1,3,1,2,3)\\
    &+(1,3,1,3,1,2,3,2,3)+(1,3,1,2,3,2,3,2,3)+(1,2,3,2,3,2,3,2,3)~.
\end{eqs}
When $r$ is a prime, we can define the Steenrod cup-$(r,i)$ product between $\ZZ_r$ coefficient cochains
\begin{eqs}
    D_i^r(B)=\phi\left(\psi(r)(e_i)\right)(B,B,\cdots,B)~,
\end{eqs}
where each of the $r$ inputs takes the value $B$. For an odd prime $r$, the Steenrod reduced powers 
\begin{eqs}
    P^s:H^q(M,\ZZ_r)\to H^{q+2s(r-1)}(M,\ZZ_r)
\end{eqs}
have the cochain level expression
\begin{eqs}
    P^s=(-1)^s\nu(q)D^r_{(q-2s)(r-1)}~,
\end{eqs}
where $\nu(q)=(-1)^{q(q-1)m/2}(m!)^q$ and $m=(r-1)/2$. 

We are interested in the case where $r=3$ and $s=1$, and the Steenrod reduced power is
\begin{eqs}
    P^1(B_q)=(-1)^{q(q-1)/2+1}D^3_{2(q-2)}(B_q)~.\label{P1D3}
\end{eqs}
Using Eq.~\eqref{P1D3}, explicit expressions for $P^1(B_2)$ and $P^1(B_3)$ can be obtained as illustrative examples: 
\begin{eqs}
    P^1(B_2)(\langle0,1,\cdots,6\rangle) &= \bigl(B_2\cup B_2\cup B_2\bigr)(\langle0,1,\cdots,6\rangle) \\
    &= B_2(\langle0,1,2\rangle)~B_2(\langle2,3,4\rangle)~B_2(\langle4,5,6\rangle)~,
\end{eqs}
and
\begin{eqs}
    &P^1(B_3)(\langle0,1,\cdots,7\rangle) \\
    =&~~B_3(\langle0,3,4,5\rangle)~B_3(\langle0,5,6,7\rangle)~B_3(\langle0,1,2,3\rangle)
    +B_3(\langle0,4,5,6\rangle)~B_3(\langle0,1,6,7\rangle)~B_3(\langle1,2,3,4\rangle)\\
    &+B_3(\langle0,5,6,7\rangle)~B_3(\langle0,1,2,7\rangle)~B_3(\langle2,3,4,5\rangle)
    +B_3(\langle0,1,4,5\rangle)~B_3(\langle1,5,6,7\rangle)~B_3(\langle1,2,3,4\rangle)\\
    &-B_3(\langle0,1,5,6\rangle)~B_3(\langle1,2,6,7\rangle)~B_3(\langle2,3,4,5\rangle)
    +B_3(\langle0,1,6,7\rangle)~B_3(\langle1,2,3,7\rangle)~B_3(\langle3,4,5,6\rangle)\\
    &+B_3(\langle0,1,2,5\rangle)~B_3(\langle2,5,6,7\rangle)~B_3(\langle2,3,4,5\rangle)
    +B_3(\langle0,1,2,6\rangle)~B_3(\langle2,3,6,7\rangle)~B_3(\langle3,4,5,6\rangle)\\
    &+B_3(\langle0,1,2,7\rangle)~B_3(\langle2,3,4,7\rangle)~B_3(\langle4,5,6,7\rangle)-B_3(\langle0,1,2,3\rangle)~B_3(\langle3,4,5,6\rangle)~B_3(\langle0,1,6,7\rangle)\\
    &-B_3(\langle0,2,3,4\rangle)~B_3(\langle4,5,6,7\rangle)~B_3(\langle0,1,2,7\rangle)-B_3(\langle0,1,2,3\rangle)~B_3(\langle3,4,5,6\rangle)~B_3(\langle1,2,6,7\rangle)\\
    &+B_3(\langle0,1,3,4\rangle)~B_3(\langle4,5,6,7\rangle)~B_3(\langle1,2,3,7\rangle)-B_3(\langle0,1,2,3\rangle)~B_3(\langle3,4,5,6\rangle)~B_3(\langle2,3,6,7\rangle)\\
    &-B_3(\langle0,1,2,4\rangle)~B_3(\langle4,5,6,7\rangle)~B_3(\langle2,3,4,7\rangle)-B_3(\langle0,1,2,3\rangle)~B_3(\langle3,4,5,6\rangle)~B_3(\langle3,4,6,7\rangle)\\
    &+B_3(\langle0,1,2,3\rangle)~B_3(\langle3,5,6,7\rangle)~B_3(\langle3,4,5,7\rangle)-B_3(\langle0,1,2,3\rangle)~B_3(\langle3,4,5,6\rangle)~B_3(\langle4,5,6,7\rangle)\\
    &-B_3(\langle0,1,2,3\rangle)~B_3(\langle3,4,6,7\rangle)~B_3(\langle4,5,6,7\rangle)~,
\end{eqs}
which has been shown in Ref.~\cite{Medina2021CochainMaySteenrod}. 

\section{Fusion group change and dimensional reduction}\label{app:Fusion group change and dimensional reduction}

Statistical processes and anomalies can be viewed as dual objects, jointly determining a phase factor in $\RR/\ZZ \cong U(1)$. Concretely, anomalies (cocycles) define the excitations on the boundary of SPT phases, and substituting the corresponding unitary operators into the statistical process yields the associated phase factor.
As discussed above, a statistical process for $\ZZ$ membranes naturally induces one for $\ZZ_N$ membranes, and a process defined in $d$ spatial dimensions can be embedded into higher dimensions as a (possibly trivial) process.
From the dual perspective, these relations can be reversed. Specifically, for a lattice system with $\ZZ_N$ membrane excitations in $d$ dimensions, one may  
\begin{enumerate}
    \item regard the fusion group as $\ZZ$ by neglecting the relation that $N$ excitations fuse to the vacuum;  
    \item reinterpret the system as $(d-1)$-dimensional by restricting to excitations and operators within a $(d-1)$-dimensional subspace.  
\end{enumerate}
Both procedures have natural interpretations in algebraic topology.  
First, the surjection $\ZZ \to \ZZ_N$ induces a map  
\begin{equation}
    H^{d+2}(B^{d-2}\ZZ_N,\RR/\ZZ)\;\longrightarrow\; H^{d+2}(B^{d-2}\ZZ,\RR/\ZZ)~,
\end{equation}
which corresponds to changing the fusion group.  
Second, the spectrum map $\Sigma K(G,d-3)\to K(G,d-2)$ induces a canonical map  
\begin{equation}
    H^{d+2}(B^{d-2}G,\RR/\ZZ)\;\longrightarrow\; H^{d+1}(B^{d-3}G,\RR/\ZZ)~,
\end{equation}
which corresponds to dimensional reduction. While these identifications remain conjectural, our explicit computations of the induced maps between generalized statistics are consistent with these expectations and thereby provide supporting evidence.  

These relationships are summarized in Fig.~\ref{fig:Z3_Z_cohomology_mapping} and Fig.~\ref{fig:Zp_Z_cohomology_mapping}, which present commutative diagrams for $\ZZ$ and $\ZZ_p$ membranes across different dimensions. For any group element, tracing the arrows to the {\color{red}$\RR/\ZZ$} entry in the upper right yields the corresponding phase factor of $\mu_{56}$. We can clearly see from the diagrams that only the statistics of $\ZZ$- and $\ZZ_{3N}$-membranes are detectable in $d\geq5$ spatial dimensions.

\begin{figure}
\centering
\begin{adjustbox}{scale=0.9}
\begin{tikzpicture}[>=Latex]
\tikzset{lab/.style={midway, font=\scriptsize, inner sep=0.3ex}}
\matrix (m) [matrix of math nodes, row sep=1.5em, column sep=2em, nodes={anchor=center}]{
  d  & |(h7)| 7 & |(h6)| 6 & |(h5)| 5 & |(h4)| 4  \\
 |(opT)| H^{d-2}(\cdot,\ZZ)\rightarrow H^{d+2}(\cdot,\RR/\ZZ)& |(t7)| \tfrac13 P^1(B_5)~;~ \tfrac12 Sq^4(B_5) & |(t6)| \tfrac13 P^1(B_4)~;~ \lambda B_4^{2} & |(t5)| \tfrac13 P^1(B_3) & |(t4)| \lambda B_2^{3} \\
 |(lab1)| H^{d+2}\!\big(B^{d-2}\mathbb Z,\ \mathbb R/\mathbb Z\big) &
 |(U7)| \mathbb Z_{3}\times\mathbb Z_{2} &
 |(U6)| \mathbb Z_{3}\times\mathbb R/\mathbb Z &
 |(U5)| \mathbb Z_{3} &
 |(U4)| {\color{red}\mathbb R/\mathbb Z } \\
 |(lab2)| H^{d+2}\!\big(B^{d-2}\mathbb Z_{3},\ \mathbb R/\mathbb Z\big) &
 |(L7)| \ZZ_3 &
 |(L6)| \ZZ_3\times \ZZ_3 &
 |(L5)| \ZZ_3\times \ZZ_3 &
 |(L4)| \ZZ_{9}  \\
 |(opB)| H^{d-2}(\cdot,\ZZ_3)\rightarrow H^{d+2}(\cdot,\RR/\ZZ) &
 |(b7)| \frac{1}{3} P^1(B_5) &
 |(b6)| \frac{1}{3}P^1(B_4)~;~\tfrac13 B_4^{2} &
 |(b5)| \frac{1}{3}P^1(B_3)~;~\tfrac19 B_3\cup \delta B_3 &
 |(b4)| \text{Pontryagin}  \\
};
\draw (-9,1.75)-- (9,1.75);   
\draw (-4.5,-3) -- (-4.5,3);  
\path[->] (U7) edge node[lab,above]{$(1,0)~;~\left(0,\tfrac12\right)$} (U6);
\path[->] (U6) edge node[lab,above]{$1~;~0$} (U5);
\path[->] (U5) edge node[lab,above]{$\tfrac13$} (U4);

\path[->] (L7) edge node[lab,above]{$(1,0)$} (L6);
\path[->] (L6) edge node[lab,above]{$(1,0)~;~(0,0)$} (L5);
\path[->] (L5) edge node[lab,above]{$3~;~0$} (L4);

\path[->] (L7) edge node[lab,left]{$(1,0)$} (U7);
\path[->] (L6) edge node[lab,left]{$(1,0)~;~(0,\tfrac13)$} (U6);
\path[->] (L5) edge node[lab,left]{$1~;~0$} (U5);
\path[->] (L4) edge node[lab,left]{$\tfrac19$} (U4);
\end{tikzpicture}
\end{adjustbox}
    \caption{Relevant cohomology groups of $K(\ZZ,d-2)$ and $K(\ZZ_3,d-2)$ and the maps between them. Generators are indicated above or below each group in terms of cohomology operations. We take the generator $1$ for $\ZZ_n$, and $(1,0);(0,1)$ for $\ZZ_m\times \ZZ_n$. Since $\RR/\ZZ$ is not finitely generated, its elements are written with a parameter $\lambda \in \RR/\ZZ$.
    Here $B_n$ denotes an $n$-cocycle, $P^1$ is the first Steenrod reduced power for $p=3$, $Sq^i$ denotes the Steenrod square operation, and ``Pontryagin'' refers to the higher Pontryagin power defined in Eq.~\eqref{eq:higher Pontryagin power}.
    Horizontal arrows represent homomorphisms induced by the spectrum map, and vertical arrows those induced by the change of fusion group. Each arrow is annotated with the images of the corresponding generators. For instance, the label $3;0$ indicates that $\tfrac{1}{3}P^1(B_3)$ maps to three times the Pontryagin operation while $\tfrac{1}{9}B_3\cup\delta B_3$ maps to zero; the label $(1,0);(0,\tfrac{1}{2})$ indicates that $\tfrac{1}{3}P^1(B_5)$ maps to $\tfrac{1}{3}P^1(B_4)$, while $\tfrac{1}{2}B_5\cup_1 B_5$ maps to $\tfrac{1}{2}B_4^2$. For any cohomology class in this commutative diagram, the value of $\mu_{56}$ is obtained from its image in the upper-right $\RR/\ZZ$.}
    \label{fig:Z3_Z_cohomology_mapping}
\end{figure}

\begin{figure}[htb]
\centering
\begin{adjustbox}{scale=0.9}
\begin{tikzpicture}[>=Latex]
\tikzset{lab/.style={midway, font=\scriptsize, inner sep=0.3ex}}
\matrix (m) [matrix of math nodes, row sep=1.5em, column sep=3em, nodes={anchor=center}]{
  d & |(h7)| 7 & |(h6)| 6 & |(h5)| 5 & |(h4)| 4 \\
 |(opT)| H^{d-2}(\cdot,\ZZ)\rightarrow H^{d+2}(\cdot,\RR/\ZZ)& |(t7)| \tfrac13 P^1(B_5)~;~ \tfrac12 Sq^4(B_5) & |(t6)| \tfrac13 P^1(B_4)~;~ \lambda B_4^{2} & |(t5)| \tfrac13 P^1(B_3) & |(t4)| \lambda B_2^{3} \\
 |(lab1)| H^{d+2}\!\big(B^{d-2}\mathbb Z,\ \mathbb R/\mathbb Z\big) &
 |(U7)| \mathbb Z_{3}\times\mathbb Z_{2} &
 |(U6)| \mathbb Z_{3}\times\mathbb R/\mathbb Z &
 |(U5)| \mathbb Z_{3} &
 |(U4)| {\color{red}\mathbb R/\mathbb Z }\\
 |(lab2)| H^{d+2}\!\big(B^{d-2}\mathbb Z_{2},\ \mathbb R/\mathbb Z\big) &
 |(L7)| \ZZ_2 &
 |(L6)| \ZZ_4 &
 |(L5)| \ZZ_2 &
 |(L4)| \ZZ_2 \\
 |(opB)| H^{d-2}(\cdot,\ZZ_2)\rightarrow H^{d+2}(\cdot,\RR/\ZZ) &
 |(b7)| \tfrac12 Sq^4(B_5) &
 |(b6)| \text{Pontryagin square} &
 |(b5)|  \tfrac14\, B_3\cup \delta B_3 &
 |(b4)| \tfrac12 B_2^3 \\
};
\draw (-9,1.75)-- (9,1.75);   
\draw (-4,-2.5) -- (-4,3);  
\path[->] (U7) edge node[lab,above]{$(1,0)\;;\; \!\left(0,\tfrac12\right)$} (U6);
\path[->] (U6) edge node[lab,above]{$1\;;\;0$} (U5);
\path[->] (U5) edge node[lab,above]{$\tfrac13$} (U4);
\path[->] (L7) edge node[lab,above]{$2$} (L6);
\path[->] (L6) edge node[lab,above]{$1$} (L5);
\path[->] (L5) edge node[lab,above]{$0$} (L4);
\path[->] (L7) edge node[lab,left]{$(0,1)$} (U7);
\path[->] (L6) edge node[lab,left]{$(0,\tfrac14)$} (U6);
\path[->] (L5) edge node[lab,left]{$0$} (U5);
\path[->] (L4) edge node[lab,left]{$\tfrac12$} (U4);
\end{tikzpicture}
\end{adjustbox}

\vspace{-0.24cm}
\begin{adjustbox}{scale=0.9}
\begin{tikzpicture}[>=Latex]
\tikzset{lab/.style={midway, font=\scriptsize, inner sep=0.3ex}}
\matrix (m) [matrix of math nodes, row sep=1.5em, column sep=3em, nodes={anchor=center}]{
 |(opT)| H^{d-2}(\cdot,\ZZ)\rightarrow H^{d+2}(\cdot,\RR/\ZZ)& |(t7)| \tfrac13 P^1(B_5)~;~ \tfrac12 Sq^4(B_5) & |(t6)| \tfrac13 P^1(B_4)~;~ \lambda B_4^{2} & |(t5)| \tfrac13 P^1(B_3) & |(t4)| \lambda B_2^{3} \\
 |(lab1)| H^{d+2}\!\big(B^{d-2}\mathbb Z,\ \mathbb R/\mathbb Z\big) &
 |(U7)| \mathbb Z_{3}\times\mathbb Z_{2} &
 |(U6)| \mathbb Z_{3}\times\mathbb R/\mathbb Z &
 |(U5)| \mathbb Z_{3} &
 |(U4)| {\color{red}\mathbb R/\mathbb Z }\\
 |(lab2)| H^{d+2}\!\big(B^{d-2}\mathbb Z_{p},\ \mathbb R/\mathbb Z\big) &
 |(L7)| 0 &
 |(L6)| \mathbb Z_{p} &
 |(L5)| \mathbb Z_{p} &
 |(L4)| \mathbb Z_{p} \\
 |(opB)| H^{d-2}(\cdot,\ZZ_p)\rightarrow H^{d+2}(\cdot,\RR/\ZZ) &
 |(b7)| 0 &
 |(b6)| \tfrac1p B_4^{2} &
 |(b5)| \tfrac1{p^2}\, B_3\cup \delta B_3 &
 |(b4)| \tfrac1p B_2^{3} \\
};
\draw (-9,2.25)-- (9,2.25);   
\draw (-4,-2.25) -- (-4,2.5);  
\path[->] (U7) edge node[lab,above]{$(1,0)\;;\; \!\left(0,\tfrac12\right)$} (U6);
\path[->] (U6) edge node[lab,above]{$1\;;\;0$} (U5);
\path[->] (U5) edge node[lab,above]{$\tfrac13$} (U4);
\path[->] (L7) edge (L6);
\path[->] (L6) edge node[lab,above]{$0$} (L5);
\path[->] (L5) edge node[lab,above]{$0$} (L4);
\path[->] (L7) edge node[lab,left]{$0$} (U7);
\path[->] (L6) edge node[lab,left]{$\left(0,\tfrac1p\right)$} (U6);
\path[->] (L5) edge node[lab,left]{$0$} (U5);
\path[->] (L4) edge node[lab,left]{$\tfrac1p$} (U4);
\end{tikzpicture}
\end{adjustbox}
    \caption{Analogous to Fig.~\ref{fig:Z3_Z_cohomology_mapping}, this diagram shows the relevant cohomology groups of $K(\ZZ, d-2)$ and $K(\ZZ_p, d-2)$ for prime $p=2$ or $p \geq 5$. In all cases with $d \geq 5$, the groups $H^{d+2}(B^{d-2}\ZZ_p,\RR/\ZZ)$ map trivially to the upper-right $\RR/\ZZ$. Consequently, $\mu_{56}$ detects nontrivial statistics in dimensions $\geq 5$ only for $G=\ZZ$ or $\ZZ_{3N}$.}
    \label{fig:Zp_Z_cohomology_mapping}
\end{figure}

\section{Algorithm for verifying membrane statistics}\label{sec:Algorithm for verifying membrane statistics}

It is known that a topological quantum field theory has a 't Hooft anomaly if there is a global symmetry that cannot be gauged. The 't Hooft anomaly naturally emerges on the boundary of an SPT phase, which makes it useful in characterizing SPT phases \cite{witten2016anomaly,chen2013groupcohomology}. For higher-form SPT, its boundary field theory will have a higher-form 't Hooft anomaly instead. Recent progress \cite{Kapustin2017Higher,lokman2020lattice} has found a characterization of higher-form SPT phases using the cohomology of Eilenberg-MacLane spaces. However, it is less well known how to obtain the boundary lattice theory from a given higher-form SPT phase. In the Supplemental Material, we describe a practical way to obtain the boundary topological phase, together with its anomalous symmetry and excitation operators, from a given higher-form SPT phase. In addition, we propose an algorithm that calculates the 56-step process of hopping operators in the SPT boundary. By implementing this algorithm on the computer, we can numerically verify that the 56-step process correctly detects the statistics of both the Pontryagin power in (5+1)D and the Steenrod reduced power in higher dimensions. 

\subsection{Simplicial formulation of higher-form SPT phase}

Throughout this section, $D$ denotes the bulk spacetime dimension, while $n$ denotes the degree of the closed background cochain $B\in Z^n(\Sigma;G)$; the protecting symmetry is an $(n-1)$-form symmetry. The manifolds $\Sigma$, $M$, and $N=\partial M$ have dimensions $D$, $D-1$, and $D-2$, respectively. The cochains $b$ and $h$ have degree $n-1$, while $\eps$ has degree $n-2$ and, when present, $\eta$ has degree $n-3$. The local functionals $T$, $\Phi$ and $\Theta$, and $F$ and $L$ have degrees $D$, $D-1$, and $D-2$, respectively. For the membrane excitations studied here, the boundary spatial dimension is $d=D-2$ and $n=d-2=D-4$.

Firstly, we introduce a concrete way of representing an SPT phase. Consider a higher-form SPT phase characterized by an element in $H^D(K(G,n),\RR/\ZZ)$. According to the simplicial construction of $K(G,n)$, we find that this cohomology element could be realized at the cochain level, that is, we have a map
\begin{eqs}
    T: C^n(\Delta^D;G)\longrightarrow C^D(\Delta^D;\RR/\ZZ) 
\end{eqs}
in which $\Delta^D$ stands for any $D$-simplex. This operation could be made universal by performing it on each $D$-simplex of a simplicial complex $\Sigma$. As a result, we obtain a local operation
\begin{eqs}
    T: C^n(\Sigma;G)\longrightarrow C^D(\Sigma;\RR/\ZZ)
\end{eqs}
that maps an $n$-form $B$ to a $D$-form $T[B]$ on $\Sigma$. 

As $T[B]$ represents a cohomology operation, it has the property that $\delta T[B]=0$ whenever $B$ is a closed form. In addition, we have that $T[B]$ and $T[B+\delta b]$ represent the same cohomology class for any $(n-1)$-form $b$. When $\Sigma$ is a triangulated manifold, the integration of $T[B]$ on $\Sigma$ is understood as the SPT action. 

Our choice of the cochain-level operation $T$ has certain ambiguities in that we can modify $T$ by a full differential without changing the SPT phase it represents. Therefore, we have an equivalence principle 
\begin{eqs}
    T[B]\sim T[B]+\delta \Psi[B]
\end{eqs}
for any $\Psi:C^{n}(\Delta^{D-1};G)\longrightarrow C^{D-1}(\Delta^{D-1};\RR/\ZZ)$~. 

\subsection{Gauge invariant field theory on triangulated spacetime}

Now, let us describe the boundary theory of a given higher SPT phase described by a local operation $T[B]$. In general, this boundary theory will not have a local action of $B$. However, if we take $B=\delta b$, then we can describe the boundary theory using a local action of $b$. Specifically, we construct a local operation 
\begin{eqs}
    \Phi:C^{n-1}(M;G)\longrightarrow C^{D-1}(M;\RR/\ZZ)
\end{eqs}
such that
\begin{eqs}\label{basic1}
    T[\delta b]=\delta \Phi[b]~.
\end{eqs}
We describe a specific construction of $\Phi$. Given a $(D-1)$-simplex $\Delta^{D-1}$ together with $b$ on it, form its cone $C\Delta^{D-1}\cong\Delta^D$. Place the cone vertex before all base vertices and orient cones by $\partial(C\sigma)=\sigma-C(\partial\sigma)$. Extend $b$ to $\tilde b$ by setting $\tilde b=b$ on the base and $\tilde b=0$ on simplices containing the cone vertex. We then define
\begin{eqs}\label{tphi}
    \langle\Phi[b], \Delta^{D-1}\rangle=\langle T[\delta \tilde{b}], C\Delta^{D-1}\rangle. 
\end{eqs}
To verify Eq.~\eqref{basic1}, it suffices to evaluate it on an arbitrary oriented $D$-simplex $\sigma$. Extend $b$ to the auxiliary $(D+1)$-simplex $C\sigma$ by the same zero-extension rule. Since $\delta\tilde b$ is closed, the cocycle property of $T$ gives
\begin{eqs}
    0
    &=\langle\delta T[\delta\tilde b],C\sigma\rangle\\
    &=\langle T[\delta b],\sigma\rangle
      -\langle T[\delta\tilde b],C(\partial\sigma)\rangle\\
    &=\langle T[\delta b]-\delta\Phi[b],\sigma\rangle~.
\end{eqs}
The last equality follows from Eq.~\eqref{tphi} on each face of $\sigma$. This proves Eq.~\eqref{basic1}. We call this procedure the cone construction.

The functional $\Phi[b]$ has a clear physical meaning. First, it describes the ground state of the SPT phase, 
\begin{eqs}
    \ket{\Psi_{\mathrm{SPT}}}=\sum_{b}\exp(2\pi i\int_M \Phi[b])\ket{b}~,
\end{eqs}
in which $b$ represents the physical degrees of freedom on the $(D-1)$-manifold $M$.
Alternatively, one may view $M$ as the background spacetime manifold of a gauge-invariant topological field theory with action 
\begin{eqs}
    S_M[b]=\int_M \Phi[b]~.
\end{eqs}
On a closed $M$, this action is invariant under the gauge transformation $b\mapsto b+\delta\eps$.
When the spacetime manifold $M$ has a spatial boundary $N=\partial M$, we can obtain a topological phase on $N$.
In this case, the $\Phi[b]$ for closed $b$ gives the properties of its gapped ground state, while in general $\Phi[b]$ gives the properties of low-energy excitations supported on $B=\delta b$. We will elaborate further on this point in the following subsections. 

Our choice of $\Phi$ also has ambiguities, and different $\Phi$ will describe the same topological phase on $N$ according to the equivalence principle below: 
\begin{eqs}
    \Phi[b]\sim \Phi[b]+\Psi[B]+\delta \Gamma[b]
\end{eqs}
for any $\Psi:C^{n}(\Delta^{D-1};G)\longrightarrow C^{D-1}(\Delta^{D-1};\RR/\ZZ)$ and $\Gamma:C^{n-1}(\Delta^{D-2};G)\longrightarrow C^{D-2}(\Delta^{D-2};\RR/\ZZ)$. 

\subsection{Stabilizer models and anomalous symmetry}

Now we are equipped with a gauge invariant field theory $\Phi[b]$ on a triangulated spacetime manifold $M$. However, to apply the statistical process, it is better to derive a Hamiltonian model for the topological phase on the boundary spatial manifold $N=\partial M$. We take a cobordism point of view, in which quantum states are defined on $(D-2)$-dimensional spatial manifolds, and the evolution is decided by a $(D-1)$-dimensional spacetime manifold between two spatial manifolds. Because we only care about local properties, we choose the spatial manifold $N$ to be $S^{D-2}=\partial B^{D-1}$, where the spacetime manifold $M=B^{D-1}$ is the $(D-1)$-dimensional ball. In this case, there is a unique ground state $\ket{\Psi_T}$ for the topological phase. 

We determine the ground state using a discrete path integral of all closed classical configurations $b$ on the spacetime manifold. Formally, we have
\begin{eqs}
    \ket{\Psi_T}=\sum_{\delta b=0} \exp(2\pi i \int_{M} \Phi[b])\ket{b_{N}}~,
\end{eqs}
in which $b_{N}$ is the restriction of $b$ on $N$. To define the stabilizers, we first introduce two types of operators. The first one is the $Z$ operator $Z(\alpha)$ defined for each chain $\alpha\in C_{n-1}(N,G)$. For each generator $e_j$ of $\ZZ_l\subset G$ and simplex $\Delta\subseteq N$, we get $Z(e_j\Delta)$ by putting an ordinary $l$-qudit $Z$ operator on $\Delta$. In general we require that $Z(\alpha)$ is a homomorphism, that is,
\begin{eqs}
    Z\!\left(\sum_{j,k}\alpha_{jk}e_j\Delta_k\right)=\prod_{j,k} Z(e_j\Delta_k)^{\alpha_{jk}}~.
\end{eqs}
Another one is the $X$ operator $X(b)$ for each cochain $b\in C^{n-1}(N,G)$, satisfying 
\begin{eqs}
    X(b)\ket{b_0}=\ket{b+b_0}~.
\end{eqs}
We would like to construct mutually commuting stabilizers that have $\ket{\Psi_T}$ as their common eigenstate with eigenvalue 1. One way to do this is to restrict to the subspace of all closed configurations, and then impose phase difference conditions among those configurations. In the following, we describe how to obtain the stabilizers from this method. 

To impose the condition $\delta b=0$, we use the $Z$-plaquette operators $Z(\partial \gamma)$, where $\gamma\in C_n(N,G)$ takes value from all the generators of $C_n(N,G)$. The condition $Z(\partial \gamma)=1$ implies that $\langle b,\partial\gamma\rangle=0=\langle\delta b,\gamma\rangle$ for arbitrary $\gamma$, which gives $\delta b=0$. Now, the difference of any two closed configurations is exact, so we need to compute the phase difference between $\ket{b}$ and $\ket{b+\delta\eps}$ for $\eps\in C^{n-2}(N,G)$. Suppose that we have fixed an extension of $b$ and $\eps$ into $M$, then the phase difference is
\begin{eqs}
    e^{i\Delta\theta}=\exp\!\left(2\pi i\int_M\bigl(\Phi[b+\delta\eps]-\Phi[b]\bigr)\right)~.
\end{eqs}
We claim that the integral value depends only on the configuration of $b$ and $\eps$ at the boundary. To prove this, we glue two balls $B_+^{D-1}$ and $B_-^{D-1}$ together along their boundaries to form a sphere $S^{D-1}$. Using the gauge invariance of $\Phi[b]$, we obtain
\begin{eqs}
    \int_{B_+^{D-1}\cup B_-^{D-1}} \bigl(\Phi[b+\delta\eps]-\Phi[b]\bigr)=0~, 
\end{eqs}
which implies that
\begin{eqs}
    \int_{B_+^{D-1}} \bigl(\Phi[b+\delta\eps]-\Phi[b]\bigr)=-\int_{B_-^{D-1}} \bigl(\Phi[b+\delta\eps]-\Phi[b]\bigr)~. \label{F1 equals to F2}
\end{eqs}
Therefore, both sides of Eq.~\eqref{F1 equals to F2} depend on the value of $b$ and $\eps$ at the boundary $B_+^{D-1}\cap B_-^{D-1}=S^{D-2}$. 

From the above argument, we find that there is a local functional $F[b,\eps]$ satisfying 
\begin{eqs}
    \int_M\bigl(\Phi[b+\delta\eps]-\Phi[b]\bigr)=\int_N F[b,\eps]~. \label{dPhi eq F}
\end{eqs}
According to Eq.~\eqref{dPhi eq F}, different closed configurations are related by $X$-type stabilizers of the form
\begin{eqs}
    W(\eps)=X(\delta\eps)\Lambda(\eps)~,\label{W definition}
\end{eqs}
where $\Lambda(\eps)$ is a diagonal operator satisfying
\begin{eqs}
    \Lambda(\eps)\ket{b}=\exp(2\pi i\int_N F[b,\eps])\ket{b}~.
\end{eqs}
Note that the $W$ operators satisfy the condition 
\begin{equation}
    W(\eps_1+\eps_2)=W(\eps_1)W(\eps_2)
\end{equation}
and
\begin{eqs}
    W(\delta\eta)=\mathbb{I}
\end{eqs}
for all $\eta\in C^{n-3}(N,G)$. This means that the $W$ operators can be seen as an anomalous $(n-1)$-form symmetry on $N$. Its anomaly is characterized by the SPT phase $T[B]$. 

At the cochain level, we choose a local representative $F$ satisfying the descent equation
\begin{eqs}
    \Phi[b+\delta\eps]-\Phi[b]=\delta F[b,\eps]~. \label{dPhi eq F diff}
\end{eqs}
Integrating this identity and applying Stokes' theorem gives Eq.~\eqref{dPhi eq F}. This form is useful when $\Phi[b]$ is constructed with cup products. 

\subsection{Excitation operators}\label{subsec:Excitation operators}

The low-energy excitations of a lattice theory are created by particular patch operators that commute with the stabilizers inside their support regions and violate some of them at the boundary. The type of excitation is determined by the class of stabilizers that its patch operator violates. For the boundary topological phase of an SPT that we constructed above, there are two types of excitation, namely, excitations created by hopping operators that violate $Z$-type stabilizers, and flux excitations that violate $X$-type stabilizers. 

It is easily seen that the flux excitation operators are just $Z(\alpha)$ with $\alpha\in C_{n-1}(N,G)$, which violates $X$-type stabilizers on $\partial \alpha$. Therefore, these operators generate bosonic excitations that have trivial statistics. On the other hand, the hopping operators are more complicated and have non-trivial statistics. In the following, we describe a systematic way of constructing them. 

Suppose that the operators supported on $h\in C^{n-1}(N,G)$ have the general form of 
\begin{eqs}
    U_h\ket{b}=\exp(2\pi i\int_N L[b,h])\ket{b+h}~\label{U definition}
\end{eqs}
in which $L[b,h]$ is a local functional to be determined. We expect $U_h$ to create excitations on $\delta h$, which means that
\begin{eqs}
    U_hW(\epsilon)=W(\epsilon)U_h~\label{U commu with W}
\end{eqs}
for all $\eps\in C^{n-2}(N,G)$. Substituting Eqs.~\eqref{W definition} and \eqref{U definition} into Eq.~\eqref{U commu with W}, we obtain 
\begin{eqs}
    \int_N\bigl(L[b,h]+F[b+h,\eps]-L[b+\delta\eps,h]-F[b,\eps]\bigr)=0~.\label{LF integral}
\end{eqs}
A sufficient local condition for Eq.~\eqref{LF integral} is
\begin{eqs}
    \delta L[b,h]+\delta F[b+h,\eps]-\delta L[b+\delta\eps,h]-\delta F[b,\eps]=0~.\label{LF diff}
\end{eqs}
Indeed, integrating this equation over $M$ and applying Stokes' theorem yields Eq.~\eqref{LF integral} on $N=\partial M$. Combining Eqs.~\eqref{dPhi eq F diff} and \eqref{LF diff}, 
\begin{eqs}
    \delta L[b,h]-\Phi[b+h]+\Phi[b]=\delta L[b+\delta\eps,h]-\Phi[b+h+\delta\eps]+\Phi[b+\delta\eps]~.\label{dL min Phi}
\end{eqs}
Eq.~\eqref{dL min Phi} implies that the value of $\delta L[b,h]-\Phi[b+h]+\Phi[b]$ is invariant with regard to the transformation $b\mapsto b+\delta \eps$. As a result, there is a local functional $\Theta[B,h]$ satisfying
\begin{eqs}
    \delta L[b,h]-\Phi[b+h]+\Phi[b]=-\Theta[\delta b,h]~.\label{def of L}
\end{eqs}
Taking the differential of both sides, we find that 
\begin{eqs}
    \delta\Theta[B,h]=T[B+\delta h]-T[B]~.\label{def of Theta}
\end{eqs}
Finally, the excitation operators are defined by Eqs.~\eqref{U definition}, \eqref{def of L}, and \eqref{def of Theta}. 

While we can derive $L[b,h]$ from Eq.~\eqref{def of L} by cone construction, the same method cannot be applied to $\Theta[B,h]$ due to the restriction $\delta B=0$. One way of constructing $\Theta[B,h]$ from Eq.~\eqref{def of Theta} is as follows. Suppose that we have fixed the value of $B$ and $h$ on a $(D-1)$-simplex $\Delta^{D-1}$, then
\begin{enumerate}
    \item Set $I=[0,1]$, with bottom surface at $0$ and top surface at $1$. Construct $\Delta^{D-1}\times I$ with its canonical triangulation, oriented by $(-1)^D$ times the ordinary product orientation. Equivalently, its oriented simplicial chain is
    \begin{eqs}
        \mathcal P_{D-1}(\Delta^{D-1})
        =\sum_{i=0}^{D-1}(-1)^{i+1}
        \langle0,\ldots,i,\bar i,\ldots,\overline{D-1}\rangle~.
    \end{eqs} 
    \item We take the pullback $\pi^* B$ of $B$ onto $\Delta^{D-1}\times I$, where $\pi:\Delta^{D-1}\times I\to \Delta^{D-1}$ is the projection. Because pullback commutes with differential, we obtain that $\pi^* B$ is closed on $\Delta^{D-1}\times I$. 
    \item Define the cochain $\iota_*h$ to equal $h$ on the bottom surface $\Delta^{D-1}\times\{0\}$ and zero on all other simplices. 
    \item The value of $\Theta$ on $\Delta^{D-1}$ is defined to be 
    \begin{eqs}
        \Theta[B,h](\Delta^{D-1})=\int_{\Delta^{D-1}\times I}T[\pi^* B+\delta(\iota_* h)]~,
    \end{eqs}
    where $\pi^* B$ and $\iota_* h$ are defined above. 
\end{enumerate}
With this orientation, $\Theta[B,h]$ satisfies Eq.~\eqref{def of Theta}. To verify this, fix an oriented $D$-simplex $\sigma$ and form the auxiliary $(D+1)$-dimensional prism $\sigma\times I$. Orient this prism by $(-1)^{D+1}$ times the ordinary product orientation, so that
\begin{eqs}
    \partial(\sigma\times I)
    =\iota_0\sigma-\iota_1\sigma
      -\mathcal P_{D-1}(\partial\sigma)~,
\end{eqs}
where $\iota_0$ and $\iota_1$ embed chains into the bottom and top surfaces, respectively. Here, $\mathcal P_{D-1}(\partial\sigma)$ is the signed sum of the prisms over the $(D-1)$-dimensional faces of $\sigma$, using the orientation specified above.
On the whole prism, set $Q=\pi^*B+\delta(\iota_*h)$, which is closed, $\delta Q=0$. Hence, the cocycle property of $T$ gives $\delta T[Q]=0$. The restrictions of $Q$ to the bottom and top are
\begin{eqs}
    \iota_0^*Q=B+\delta h~,\qquad
    \iota_1^*Q=B~.
\end{eqs}
Its restriction to each face prism agrees with the background used to define $\Theta$ on that face. Therefore,
\begin{eqs}
    \delta\Theta[B,h](\sigma)
    &=\Theta[B,h](\partial\sigma)\\
    &=\big\langle T[Q],\mathcal P_{D-1}(\partial\sigma)\big\rangle\\
    &=\big\langle T[Q],\iota_0\sigma-\iota_1\sigma\big\rangle
      -\big\langle T[Q],\partial(\sigma\times I)\big\rangle\\
    &=T[B+\delta h](\sigma)-T[B](\sigma)
      -\big\langle\delta T[Q],\sigma\times I\big\rangle\\
    &=T[B+\delta h](\sigma)-T[B](\sigma)~.
\end{eqs}
The penultimate equality uses Stokes' theorem, and the last equality follows from $\delta T[Q]=0$.

Let us calculate $\Theta[B_2,h_1]$ for $T[B_2]=\tfrac{1}{2} B_2 \cup B_2 \cup B_2$ as an example. We label the vertices of $\Delta^5\times I$ as $0,1,\cdots,5$ on the bottom surface and $\bar{0},\bar{1},\cdots,\bar{5}$ on the top surface. The canonical triangulation is given as follows: 
\begin{eqs}
    \Delta^5\times I&=\langle0,1,2,3,4,5,\bar{5}\rangle-\langle0,1,2,3,4,\bar{4},\bar{5}\rangle+\langle0,1,2,3,\bar{3},\bar{4},\bar{5}\rangle\\
    &-\langle0,1,2,\bar{2},\bar{3},\bar{4},\bar{5}\rangle+\langle0,1,\bar{1},\bar{2},\bar{3},\bar{4},\bar{5}\rangle-\langle0,\bar{0},\bar{1},\bar{2},\bar{3},\bar{4},\bar{5}\rangle~,
\end{eqs}
in which the sign of each simplex represents its orientation. Summing over $T[\pi^* B_2+\delta(\iota_* h_1)]$ for each simplex, we obtain
\begin{eqs}
    \Theta[B_2,h_1](\langle0,1,2,3,4,5\rangle)&=\frac{1}{2}\Big[(B_2+\delta h_1)(\langle0,1,2\rangle)~(B_2+\delta h_1)(\langle2,3,4\rangle)~h_1(\langle 4,5\rangle)\\
    &\qquad+(B_2+\delta h_1)(\langle0,1,2\rangle)~h_1(\langle 2,3\rangle)~B_2(\langle 3,4,5\rangle)\\
    &\qquad+h_1(\langle 0,1\rangle)~B_2(\langle1,2,3\rangle)~B_2(\langle3,4,5\rangle )\Big]~,\label{Theta012345}
\end{eqs}
which can be written in a compact form using cup products, 
\begin{eqs}
    \Theta[B_2,h_1]&=\frac{1}{2}\Big[(B_2+\delta h_1)\cup (B_2+\delta h_1)\cup h_1
    +(B_2+\delta h_1)\cup h_1\cup B_2
    +h_1\cup B_2\cup B_2\Big]~.
\end{eqs}
It is equivalent to the $\Theta$ that we could obtain by writing $T[B_2+\delta h_1]-T[B_2]$ as an exact differential. 

\subsection{Calculating the statistics of hopping operators}

To obtain the statistics of the membrane excitations on the SPT boundary, we need to substitute the hopping operators described in Eq.~\eqref{U definition} into our 56-step process. However, an explicit calculation of $L[b,h]$ tends to be complicated. We avoid this complexity by introducing an equivalent expression for $U_h$,  
\begin{eqs}
    U_h=X_h\exp \!\left( 2\pi i \int_{M} \big(-\Theta[\delta b,h]+\Phi[b+h]-\Phi[b]\big)\right)~, 
\end{eqs}
which is derived using the Stokes theorem and Eq.~\eqref{def of L}. 

Note that the term $\Phi[b+h]-\Phi[b]$ gives the difference of $\Phi$ between the final state and the initial state. When we take the sum over the 56-step process, this term vanishes because the state $\ket{b}$ goes through a closed circle. Therefore, we obtain the same result by using the modified excitation operator
\begin{eqs}
    \tilde{U}_h=X_h\exp \!\left( 2\pi i \int_{M} \big(-\Theta[\delta b,h]\big)\right)~.\label{modified excitation}
\end{eqs}
To simplify the calculation, we take the filling manifold $M$ to be a single $(D-1)$-simplex $\Delta^{D-1}$, with spatial boundary $N=\partial\Delta^{D-1}$. Under this assumption, the integral contains only one term of $\Theta[\delta b,h]$, and we can add 56 terms of $\Theta$ together to obtain the final result. The function $\Theta[B,h]$ is obtained from the SPT action $T[B]$ using the algorithm mentioned above. 

Our original definition of the 56-step process uses 2-chains for excitations. To get the process for excitations labeled by cochains, we use complementary faces of $\Delta^{D-1}$ to define a duality on its boundary $N=\partial\Delta^{D-1}$. In particular, for each $p$-chain $\langle i_0 i_1\cdots i_p\rangle$, its dual cochain is the delta function supported on the complement simplex $\langle j_0j_1\cdots j_{D-2-p}\rangle$, where both vertex sequences are increasing and together partition $\{0,1,\ldots,D-1\}$, with a sign factor $\prod_{r=0}^{p}(-1)^{i_r}=(-1)^{\sum_{r=0}^{p}i_r}$. This sign factor ensures that our duality takes the boundary operator $\partial$ of chains to the coboundary operator $\delta$ of cochains. 

In summary, the algorithm for verifying membrane statistics for $T[B]$ proceeds as follows: 
\begin{enumerate}
    \item Calculate the function $\Theta[B,h]$ from $T[B]$ using the algorithm in SM~Sec.~\ref{subsec:Excitation operators}. 
    \item Obtain the 56-step process for cochains by taking the complementary-face duality on $N=\partial\Delta^{D-1}$. 
    \item Substitute the modified excitation operators from Eq.~\eqref{modified excitation} into the process. 
    \item Verify that $\mu_{56}$ generates the statistics as expected. 
\end{enumerate}
The result is summarized in Table \ref{table: output mu56} of the main text.